\documentclass{article}
\usepackage[utf8]{inputenc} 

\usepackage[utf8]{inputenc} 
\usepackage[T1]{fontenc}    
\usepackage{hyperref}       
\usepackage{xcolor}
\usepackage{url}            
\usepackage{booktabs}       
\usepackage{amsfonts}       
\usepackage{mathtools}
\usepackage{amsmath,amssymb,amsthm,commath}
\usepackage{cleveref}
\newtheorem{theorem}{Theorem}[section]
{
\theoremstyle{definition}
\newtheorem{definition}{Definition}[section]
}
\newtheorem{lemma}{Lemma}[section]
\usepackage[accepted]{icml2021}
\usepackage{macro}
\icmltitlerunning{Dimensionality Reduction for Sum-of-Distances Metric}
\usepackage[inline, shortlabels]{enumitem}

\date{}
\renewcommand{\P}{\mathbb{P}}
\begin{document}
\twocolumn[
\icmltitle{Dimensionality Reduction for Sum-of-Distances Metric}




\icmlsetsymbol{equal}{*}

\begin{icmlauthorlist}
\icmlauthor{Zhili Feng}{cmu}
\icmlauthor{Praneeth Kacham}{cmu}
\icmlauthor{David P. Woodruff}{cmu}
\end{icmlauthorlist}

\icmlaffiliation{cmu}{ Carnegie Mellon University, Pittsburgh, USA}

\icmlcorrespondingauthor{Praneeth Kacham}{pkacham@cs.cmu.edu}

\icmlkeywords{Machine Learning, ICML}

\vskip 0.3in
]
\printAffiliationsAndNotice{}
\begin{abstract}
We give a dimensionality reduction procedure to approximate the sum of distances of a given set of $n$ points in $\R^d$ to any ``shape'' that lies in a $k$-dimensional subspace. Here, by ``shape'' we mean any set of points in $\R^d$. Our algorithm takes an input in the form of an $n \times d$ matrix $A$, where each row of $A$ denotes a data point, and outputs a subspace $P$ of dimension $O(k^{3}/\epsilon^6)$ such that the projections of each of the $n$ points onto the subspace $P$ and the distances of each of the points to the subspace $P$ are sufficient to obtain an $\epsilon$-approximation to the sum of distances to any arbitrary shape that lies in a $k$-dimensional subspace of $\R^d$. These include important problems such as $k$-median, $k$-subspace approximation, and $(j,l)$ subspace clustering with $j \cdot l \le k$. Dimensionality reduction reduces the data storage requirement to $(n+d)k^{3}/\epsilon^6$ from $\nnz(A)$.  
Here $\nnz(A)$ could potentially be as large as $nd$. Our algorithm runs in time $\nnz(A)/\epsilon^2 + (n+d)\poly(k/\epsilon)$, up to logarithmic factors. For dense matrices, where $\nnz(A) \approx nd$, we give a faster algorithm, that runs in time $nd + (n+d)\poly(k/\epsilon)$ up to logarithmic factors. Our dimensionality reduction algorithm can also be used to obtain $\poly(k/\epsilon)$ size coresets for $k$-median and $(k,1)$-subspace approximation problems in polynomial time.
\end{abstract}

\section{Introduction}
Machine learning models often require millions of high-dimensional data samples in order to train. For example, an image with moderate resolution can easily have more than a million pixels. It is crucial that we can decrease the size of the data to save on computational power. One way to decrease the size of the data is dimensionality reduction, where we project our data samples onto a low-dimensional subspace and perform the task on the low-dimensional points. Given a set of $n$ points $A = \{a_1,\ldots,a_n\}$ in $\R^{d}$, the projections of $A$ onto a subspace $P$ of $k$ dimensions needs only $k$ parameters for each point in the dataset. Thus the size of the data is proportional to $(n+d)k$, which can be much smaller than $nd$. Therefore if there exists a subspace $P$ of dimension $k$, where $k$ is much smaller than $n$ and $d$, and for which the projections of $A$ onto the subspace $P$ alone are sufficient to perform a certain a task on the dataset $A$, then we can achieve a significant reduction in the size of the data.

One very common task that requires dimensionality reduction is the shape-fitting problem. A problem instance is defined by a quadruple $(A, \mathcal{S}, \dist, f)$, where $A = \set{a_1,\ldots,a_n} \subseteq \R^d$ is a set of points, $\dist:\R^d\times \R^d \to \R_{\ge 0}$ is a metric which we will also refer to as the distance function, $\mathcal{S}$ is a collection of subsets in $\R^d$ which we call \textsl{shapes}, and a function $f : \R_{\ge 0} \rightarrow \R_{\ge 0}$. The task is to find a shape $S \in \mathcal{S}$ that minimizes $\sum_i f(\dist(a_i, S))$. The most common function $f$ used is $f(x) = x^2$ as it has a natural Frobenius norm interpretation for many tasks and has closed-form solutions for natural sets $\mathcal{S}$ of shapes. Recently, the function $f(x) = x$ has been considered as it is more robust to outliers than the function $f(x) = x^2$, meaning that it does not square the distance to an erroneous point, allowing the objective to fit more of the remaining (non-outlier) data points. 

The most common dimensionality reduction techniques include Principal Component Analysis (PCA) and the  Johnson-Lindenstrauss Transform (JL). PCA projects the original dataset onto a low-dimensional subspace for which the data variance is the largest. On the other hand, the JL transform provides a data-oblivious dimensionality reduction that preserves pairwise distances between points in the dataset.

\citet{feldman2013turning} show that if $P$ is the subspace spanned by the top $O(k/\epsilon^2)$ singular vectors of the data matrix $A$, which is given by PCA, then for any shape $S$ that lies in a $k$-dimensional space, the quantity $\sum_i \min_{s\in S}\|a_i-s\|_2^2$ can be approximated by $\sum_i \min_{s \in S}\|\P_Pa_i - s\|_2^2 + \sum_i \|a_i - \P_P a_i\|_2^2$, where $\P_Pa_i$ denotes the Euclidean projection of $a_i$ onto the subspace $P$, thereby giving a dimensionality reduction technique for the shape-fitting problem instantiated with $f(x) = x^2$, Euclidean norm distance function $\dist(x,y) = \|x- y\|_2$, and with $\mathcal{S}$ being the collection of any $k$-dimensional shape.


In this work, we concentrate on shape fitting problems with $\dist(x,y) =  \|x-y\|_2$ and $f(x) = x$. Unfortunately, both PCA and the JL transform are not known to work in this case. We give fast algorithms to find a subspace $P$ of $\widetilde{O}(k^{3}/\epsilon^6)$\footnote{We use $\widetilde{O}(f(n))$ notation to denote $O(f(n)\mathsf{polylog}(f(n)))$.} dimensions that allows us to compute a $(1 \pm \epsilon)$-approximation to $\sum_{i}\dist(x_i,S)$ for any shape $S$ that lies in a $k$-dimensional subspace. Examples of such shapes include all $k$-dimensional subspaces themselves, which corresponds to the subspace approximation problem, as well as all sets of $k$ points, which corresponds to the $k$-median problem. Our results also apply to the $(j,l)$-projective clustering problem, with $j \cdot l \le k$, where we seek to find $j$ subspaces, each of dimension at most $l$, so as to minimize the sum of distances of each input point to its nearest subspace among the $j$ that we have chosen. We also show empirically that we need fewer dimensions than our theoretical analysis predicts to obtain good approximations.

A coreset is another type of data structure to reduce the size of a data set $A$. Namely, a coreset $P$ is a data structure consuming a much smaller amount of memory than $A$, which can be used as a substitute for $A$ for any query $Y$ on $A$. For example, in the $k$-median problem, the query $Y=\{y_1,\ldots, y_k\}$ can be a set of $k$ points, and we want to find a coreset $P$ to obtain a $(1+\epsilon)$-approximation to $\sum_{i=1}^n\|a_i - y_{a_i}\|_2$, where $y_{a_i}$ is the closest point to $a_i$ in $Y$. Often, we want to construct a \textsl{strong coreset}, meaning with high probability, $P$ can be used in place of $A$ simultaneously for all possible query sets $Y$. If this is the case, then we can throw away the original dataset $A$, which saves us not only on computational power, but also on storage. 

  There is a long line of work which focuses on constructing coresets for subspace approximation with sum of squared distances loss function, as well as for the $k$-means problem (see, e.g., \cite{deshpande2006matrix, deshpande2007sampling, feldman-langberg, feldman2010coresets, feldman2013turning, varadarajan2012sensitivity, shyamalkumar2007efficient, badoiu2002approximate, chen2009coresets, feldman2012data, frahling2005coresets, frahling2008fast, har2007smaller, har2004coresets, langberg2010universal}). \citet{feldman2013turning} give the first coresets of size independent of $d$. For subspace approximation, they give strong coresets of size $O(k/\epsilon)$, and $\widetilde{O}(k^3/\epsilon^4)$  for $k$-means. \citet{cohen2015dimensionality} improve the result and give an input sparsity time algorithm to construct the coreset. 
  
  Later, \citet{sohler2018strong} give a strong coreset of size $\poly(k/\epsilon)$ for the $k$-median problem, as well as the subspace approximation problem with the sum of distances loss function, obtaining the first strong coresets independent of $n$ and $d$ for this problem. Their algorithm runs in $\widetilde O(\nnz(A)+(n+d) \cdot \poly(k/\epsilon)+\exp(\poly(k/\epsilon)))$ time. Recent work by \citet{makarychev2019performance} provides an oblivious dimensionality reduction for $k$-median to an $O(\epsilon^{-2}\log(k/\epsilon))$-dimensional space while preserving the cost of every clustering. This dimension reduction result can also be used to construct a strong coreset of size $\poly(k/\epsilon)$. 
  
\citet{sohler2018strong} gave an algorithm to compute first polynomial size coresets for $k$-median using their dimensionality reduction, albeit, with a running time exponential in $k,1/\epsilon$ as discussed. In a similar way, we can obtain $\tilde{O}(k^{4}/\epsilon^8)$ size coreset for $k$-median in polynomial time using our dimensionality reduction algorithm. In concurrent and independent work, \citet{huang-vishnoi} gave a polynomial time algorithm to compute a coreset of size $\tilde{O}(k/\epsilon^4)$. We stress that we can run the second stage in the coreset construction algorithm of \citet{huang-vishnoi} on a coreset of size $\tilde{O}(k^{4}/\epsilon^8)$ to obtain a coreset of size $\tilde{O}(k/\epsilon^4)$ just as in \cite{huang-vishnoi}. Also, their techniques cannot be extended to give an efficient dimensionality reduction algorithm to approximate the sum-of-distances metric, as their coreset construction arguments are based on an existential dimensionality reduction result. As an aside, we observe that the coreset construction algorithm of \citet{huang-vishnoi} can be implemented to have a running time of $\tilde{O}(\nnz(A) + (n+d)\poly(k/\epsilon))$, improving their $\tilde{O}(\nnz(A) \cdot k)$ time algorithm. To our knowledge, this is the first input sparsity time algorithm to construct a coreset for the fundamental $k$-median problem. We give a proof of our observation in the supplementary material.
  
The size of the coresets constructed based on the sensitivity sampling framework of \citet{feldman-langberg} depend on the dimension of the data. An important consequence of our dimensionality reduction algorithm is that the coresets constructed on the data after first reducing its dimensionality can be much smaller for many important problems.
  
\subsection{Our Results}
Our main contribution is that we obtain the first polynomial time, in fact near-linear time, dimension reduction algorithm that returns a $\poly(k/\epsilon)$-dimensional subspace such that the projections of the input points to this subspace, as well as the distances of the points to this subspace, can be used to compute a $(1 \pm \epsilon)$-approximation to the sum of distances of the set $A$ to any $k$-dimensional shape $S$.

\begin{theorem}[Dimensionality Reduction]
Given $A \in \R^{n \times d}$ and $0 < \epsilon < 1$, there exists an algorithm that runs in time $\tilde{O}(\nnz{(A)}/\epsilon^2 + (n+d)\poly(k/\epsilon))$ and outputs a subspace $P$ of dimension $\tilde{O}(k^{3}/\epsilon^6)$ such that, with probability $\ge 2/3$, for \emph{any} shape $S \subseteq \R^d$ that lies in a $k$-dimensional subspace,
\begin{equation*}
			\sum_{i}\sqrt{\dist(\P_P a_i, S)^2 + \dist(a_i, P)^2} = (1 \pm \epsilon)\sum_i \dist(a_i, S).
\end{equation*}
\label{thm:sparse-main-result-informal}
\end{theorem}	
When $A$ is dense, i.e., $\nnz(A)\approx nd$, the quantity $\nnz(A)/\epsilon^2 \approx nd/\epsilon^2$ may be prohibitive. In this case, we also provide a fast dimensionality reduction algorithm which runs in $\widetilde O(nd +(n+d) \cdot \poly(k/\epsilon))$ time.
\begin{theorem}
	For any $\epsilon\in(0, 1)$ and $k\geq 1$, there is an $\widetilde O(nd+(n+d) \cdot \poly(k/\epsilon))$ time algorithm that finds an $\tilde{O}(k^{3.5}/\epsilon^6)$-dimensional subspace $P$ such that, with probability $\ge 2/3$, for \emph{any} shape $S \subseteq \R^d$ that lies in a $k$-dimensional subspace,
\begin{equation*}
			\sum_{i}\sqrt{\dist(\P_P a_i, S)^2 + \dist(a_i, P)^2} = (1 \pm \epsilon)\sum_i \dist(a_i, S).	
\end{equation*}

\label{thm:dense-main-result-informal}
\end{theorem}
Given a subspace $P$ as in the above theorems, it is still expensive to compute the projections of the rows of $A$ onto the subspace $P$ as well as the distances to the subspace $P$. We also give an algorithm to compute approximate projections and approximate distances that still satisfy the guarantees of the above theorems, obtaining the following theorem.
\begin{theorem}[Size Reduction]
	Given a matrix $A \in \R^{n \times d}$ and a subspace $P$ of $r = \tilde{O}(k^{3}/\epsilon^{6})$ dimensions that satisfies the guarantees of Theorems~\ref{thm:sparse-main-result-informal} and \ref{thm:dense-main-result-informal}, there is an algorithm that runs in time $\tilde{O}(\nnz(A) + (n+d)\poly(k/\epsilon))$ and outputs vectors $a^{B}_i \in \R^{r}$ and values $v_i \in \R_{\ge 0}$ for all $i$ such that for any shape $S$ that lies in a $k$ dimensional subspace, 
	\begin{equation*}
		\sum_i \sqrt{\dist(Ba_i^B, S)^2 + v_i^2} = (1 \pm \epsilon)\sum_i \dist(a_i,S),
	\end{equation*}
	where $B$ is an orthonormal basis for the subspace $P$. Thus the storage requirement drops from $\nnz(A)$ to $(n+d)k^{3}/\epsilon^6$.
\end{theorem}
%

\section{Preliminaries and Technical Overview}\label{sec:prelim}
We let $A\in\R^{n\times d}$ denote our input matrix. The rows of $A$ are interpreted as a set of $n$ points in $\R^d$. Throughout the paper, we use $A_{i*}$ and $a_i$ to denote the $i^{\text{th}}$ row of $A$, and $A_{*i}$ to denote the $i^{\text{th}}$ column. Similarly, for $J \subseteq [n]$, $A_{J*}$ denotes the matrix with rows of $A$ only indexed by $J$. For $n \in \mathbb{Z}^+$, $[n]$ denotes the set $\{1,2,3,\ldots,n\}$. For a matrix $A$, we use $A^+$ to denote its Moore-Penrose pseudoinverse. We write $x = (a,b)y$ to denote that $ay \le x \le by$. If $a = 1-\epsilon$ and $b = 1 + \epsilon$, we abbreviate the notation as $x = (1 \pm \epsilon)y$.

Given a subspace $B$, we use $\P_B$ to denote the projection matrix onto $B$, i.e., for any vector $u$, we have $\P_Bu = \argmin_{v \in B}\|u-v\|_2$. Let $B^\perp$ denote the orthogonal complement of the subspace $B$. We use bold capital letters such as $\bS, \bL$ to stress that these are random matrices that are explicitly sampled.

\begin{definition}[$(p, 2)$-norm]
	For a matrix $A\in \R^{n\times d}$, its $(p,2)$-norm is  
	$
		\|A\|_{p, 2}=(\sum_{i=1}^n \|A_{i*}\|_2^p)^{1/p}.
	$ We define $\|A\|_h$ to be $\|A^\mathsf{T}\|_{1,2}$ which is the sum of $\ell_2$ norms of columns of $A$.
\end{definition}

\begin{definition}[$(k,p)$-clustering]
	Given input matrix $A\in\R^{n\times d}$, let $\mathcal X$ be the collection of all sets containing $k$ points. The $(k,p)$-clustering problem denotes the optimization problem
	$
		\min_{X\in\mathcal X}\sum_{A_{i*}\in A} d(A_{i*}, X)^p.
	$
	
	If $p=2$, we have the  $k$-means problem, while if $p=1$, we have the  $k$-median problem. 
\end{definition}

\begin{definition}[$(k,p)$-subspace approximation]
	Given input matrix $A\in\R^{n\times d}$, let $\mathcal P$ be the set of all subspaces with dimension at most $k$. The $(k,p)$-subspace approximation problem denotes the optimization problem
$
		\min_{P\in \mathcal P} \sum_{i\in[n]} d(A_{i*}, P)^p.
$
	We let $\text{SubApx}_{k,p}(A)$ denote the optimum value of the $(k,p)$ subspace approximation to $A$.
\end{definition}

\begin{definition}[$\epsilon$-strong coreset]
	For the  $(k, p)$-clustering problem with input matrix $A\in\R^{n\times d}$, a weighted $\epsilon$-strong coreset is a tuple $(C,w)$ where $C \in \R^{m \times d}$ and $w:\text{rows}(C)\to\R^+ $ is such that simultaneously for all $X \subseteq \R^d$ with $|X| = k$,
	\[
		\sum_{i\in [m]} w(C_{i*})d(C_{i*}, X)^p = (1\pm \epsilon )\sum_{i \in [n]}d(A_{i*}, X)^p.
	\]
	The definition can be generalized to \emph{any} data structure that lets us compute a $(1 \pm \epsilon)$ approximation to $\sum_{A_{i*} \in A}d(A_{i*},X)^p$ for all sets $X$ of size $k$. A similar notion of strong coreset can be defined for the $(k,p)$-subspace approximation problem as well.
\end{definition}

\begin{definition}[$(\alpha,\beta)$-bicriteria approximation]
    Given an input matrix $A \in \R^{n \times d}$ for the $(k,1)$-subspace approximation problem, we say that a subspace $Q$ is an $(\alpha, \beta)$-bicriteria approximation if $\text{dim}(Q) \le \beta$ and 
    $
        \sum_{i = 1}^n d(A_{i*},Q) \le \alpha \cdot \text{SubApx}_{k,1}(A).
    $
\end{definition}

\begin{definition}[$\ell_1$ subspace embedding]
    Let $A\in\R^{n\times d}$, $\Pi\in\R^{s\times n}$. We call $\Pi$ an $(\alpha,\beta)$ $\ell_1$ subspace embedding if for all $x\in\R^d$, $\alpha\|Ax\|_1 \le \|\Pi Ax\|_1 \le  \beta \|Ax\|_1$. If $QR = \Pi A$ is the QR decomposition, then we let $\|A_{i*}R^{-1}\|_1$ be the $\ell_1$ leverage score of the $i^{\textnormal{th}}$ row. See \cite{cohen-peng, wang-woodruff} for several constructions of $\ell_1$ subspace embeddings.
\end{definition}
\subsection{Technical Overview}
Let $A \in \R^{n \times d}$ be the input matrix. \citet{sohler2018strong} show that if a subspace $S$ satisfies 
\begin{equation}
    \|A(I-\P_S)\|_{1,2} - \|A(I-\P_{S + W})\|_{1,2} \le \epsilon^2\textnormal{SubApx}_{k,1}(A)
    \label{eqn:extension-property}
\end{equation}
for all $k$-dimensional subspaces $W$, then we can reduce the dimension of the input points by projecting the points onto $S$, while being able to compute a $(1 \pm \epsilon)$-approximation to the sum of distances to any $k$-dimensional shape. They construct such a subspace $S$ by directly computing a $(1+\epsilon,\poly(k/\epsilon))$ bicriteria approximation for the $(i^*k,1)$ subspace approximation problem on $A$, where $i^*$ is a randomly chosen index in $[1/\epsilon^2]$. This introduces the $\exp(\poly(k/\epsilon))$ term in their running time. We show that we can compute $(1+\epsilon,\poly(k/\epsilon))$-bicriteria solutions for the $(k,1)$-subspace approximation problem on $A(I-P)$, for adaptively chosen projection matrices $P$, and that with constant probability, the union of the bicriteria solutions we compute has the desired property \eqref{eqn:extension-property}.

We solve the problem of finding a $(1+\epsilon, \poly(k/\epsilon))$-bicriteria solution for the $(k,1)$-subspace approximation problem on the input $A(I-P)$, where $P$ is an arbitrary projection matrix onto a subspace of dimension at most $\poly(k/\epsilon)$, based on techniques from \cite{clarkson2015input}. We simplify their arguments and obtain tighter parameters for their algorithms. We solve the problem in two stages. First we compute an $(O(1),\tilde{O}(k))$-approximation, i.e., we find a subspace $\hat{X}$ of dimension at most $\tilde{O}(k)$ such that
$
	\|A(I-P)(I-\P_{\hat X})\|_{1,2} \leq O(1)\cdot \textnormal{SubApx}_{k,1}(A(I-P)).
$

To achieve this guarantee, we make use of so-called lopsided embeddings. \citet{clarkson2015input} show that if a matrix $S$ is an $\epsilon$ lopsided embedding for $(V_k, (A(I-P))^\mathsf{T})$, where $V_k$ is an orthonormal basis for the $k$-dimensional subspace that attains the cost $\textnormal{SubApx}_{k,1}(A(I-P))$, then $\min_{\text{rank-}k\ X}\|A(I-P)S^\mathsf{T} X - A\|_{1,2} \le (1+\epsilon)\textnormal{SubApx}_{k,1}(A(I-P))$. We first show that a Gaussian matrix $\bS$ with $O(k)$ rows is an $O(1)$ lopsided embedding with probability $\ge 9/10$. Then we show that if a random matrix $\bL$ is an $O(1)$ $\ell_1$ subspace embedding for the matrix $A(I-P){\bS}^\mathsf{T}$ and satisfies $\E_{\bL}[\|\bL M\|_{1,2}] = \|M\|_{1,2}$ for any fixed matrix $M$, then the row space of $(\bL A(I-P))$ is an $O(1)$ approximation. We use the Lewis weight sampling algorithm of \citet{cohen-peng} to sample a matrix $\bL$ that satisfies these properties. As the matrix ${\bS}^\mathsf{T}$, which is a Gaussian matrix, has only $O(k)$ columns, the matrix $\bL$ has only $\tilde{O}(k)$ rows. We can also instead use the $\ell_1$ subspace embeddings of \citet{wang-woodruff} to construct an $\tilde{O}(k^{3.5})$-sized $\ell_1$ embedding by leverage score sampling \cite{dw-sketching}.

Next, based on the $(O(1),\tilde{O}(k))$ bicriteria solution, we perform non-adaptive residual sampling. This was shown to give a $(1+\epsilon, \tilde{O}(k^{3}/\epsilon^2))$ bicriteria solution in \cite{clarkson2015input} when an $O(1)$ approximate solution is used. Thus, we obtain a subspace $\hat S$ for which 
\[
	\|A(I-P)(I-\P_{\hat S})\|_{1,2} \leq (1+\epsilon)\textnormal{SubApx}_{k,1}(A(I-P)).
\]

Starting with $P = 0$, we obtain a $(1+\epsilon,\,k^{3}/\epsilon^2)$ bicriteria subspace $\hat{S}$. However, the dimensionality reduction requires a subspace that satisfies \eqref{eqn:extension-property}. To obtain such a guarantee, we crucially run this algorithm adaptively $\Theta(1/\epsilon)$ times. Let $\hat{S}_i$ be the subspace obtained in the $i^{\text{th}}$ iteration. In the $i^{\text{th}}$ iteration, we find a bicriteria solution for the $(k,1)$ subspace approximation problem on the matrix $A(I-\P_{\hat{S}_1 \cup \ldots \cup \hat{S}_{i-1}})$. We then show that the final subspace $\hat{S} = \cup_j \hat{S}_j$, with probability $\ge 9/10$, satisfies 
$
    \|A(I-\P_{\hat{S}})\|_{1,2} - \|A(I-\P_{\hat{S} + W})\|_{1,2} \le \epsilon \cdot \textnormal{SubApx}_{k,1}(A)
$
for all $k$-dimensional subspaces $W$. Thus, running the above procedure with parameter $\epsilon^2$ gives a subspace that satisfies \eqref{eqn:extension-property}. We show that each iteration of the algorithm takes $\tilde{O}(\nnz(A) + (n+d)\poly(k/\epsilon))$ time and as we run the algorithm adaptively for $1/\epsilon^2$ iterations, the total time complexity of the algorithm is $O(\nnz{(A)}/\epsilon^2 + (n+d)\poly(k/\epsilon))$.

For dense inputs $A$, the algorithm described above has a running time of $O(nd/\epsilon^2 + (n+d)\poly(k/\epsilon))$ which can be prohibitive when both $n$ and $d$ are large. We observe that in each of the $1/\epsilon^2$ iterations, the algorithm computes sampling probabilities $p_i$ for all the $n$ rows, whereas it samples only $\poly(k/\epsilon)$ rows independently with these probabilities in any particular iteration. We propose a novel alternate sampling scheme in which we partition rows of the matrix $A(I-P)$ into equal size blocks $I_1,\ldots,I_b \subseteq [n]$. We show that given several precomputed matrices, we can quickly obtain estimates ``$\apx_j$'' that approximate $\sum_{i \in I_j}p_i$ for all $j \in [b]$. Now, to sample a row $i \in [n]$ with probability approximately equal to $p_i$, we first sample a block $I_j$ with probability proportional to $\apx_j$, which is close to $\sum_{i \in I_j}p_j,$ and then compute the probabilities $p_i$ \emph{only} for the rows in the sampled blocks. If the number of samples is less than the number of blocks, we see that we compute the actual probabilities only for a few rows. In order to be able to estimate $\apx_j$, we make use of several standard properties of Cauchy and Gaussian random matrices. Finally, we show that each of the precomputed matrices required can be computed in time $\tilde{O}(nd)$ using the fast rectangular matrix multiplication algorithm by \citet{coppersmith}.

In addition to providing a tool for data size reduction, our dimensionality reduction also leads to small coreset constructions for various problems with sizes that depend only on the problem parameter $k$ instead of $n$ or $d$. As shown by \citet{sohler2018strong}, the points projected onto the subspace given by a  dimensionality reduction algorithm can be used to construct coresets of sizes $\poly(k/\epsilon)$ for $k$-median and $(k,1)$-subspace approximation problems. We note that the same constructions work with our dimensionality reduction algorithm. We include the details of such constructions in the supplementary material.

\section{Sum of Distances to a \texorpdfstring{$k$}{k}-dimensional shape}\label{sec:sum-of-distances}
Let $A = \{a_1,\ldots,a_n\}$ be a given set of points and $P$ be a $\poly(k/\epsilon)$ dimensional subspace that satisfies \eqref{eqn:extension-property}. Let $S \subseteq \R^d$ be an arbitrary shape such that $\text{span}(S)$ has dimension at most $k$. We want to obtain an $\epsilon$ approximation to $\sum_i \dist(a_i,S)$. 

\citet{sohler2018strong} show that for any such shape $S$,
$
\sum_i \sqrt{\dist(a_i, \P_Pa_i)^2 + \dist(\P_Pa_i,S)^2} = (1 \pm \epsilon)\sum_{i \in S}\dist(a_i, S).
$
The following lemma is a more general version that works with approximate projections onto the subspace $P$ and approximate distances to the subspace $P$. A similar lemma is stated as Lemma~14 in \cite{sohler2018strong}. We correct an error in Equation~2 of their proof.
\begin{theorem}
	Let $P$ be an $r$ dimensional subspace of $\R^d$ such that
	\begin{equation*}
	    \sum_i \dist(a_i,P) - \sum_i \dist(a_i,P + W) \le \frac{\epsilon^2}{80}\textnormal{SubApx}_{k,1}(A)
    \end{equation*}
	for all $k$-dimensional subspaces $W$. Let $B \in \R^{d \times r}$ be an orthonormal basis for the subspace $P$. For each $a_i$, let $a_i^B \in \R^r$ be such that
	$
		\dist(a_i, Ba_i^B) \le (1+\epsilon_c)\dist(a_i,P)
	$
	and let $(1-\epsilon_c)\dist(a_i,P) \le \apx_i \le (1+\epsilon_c)\dist(a_i,P)$ for $\epsilon_c = \epsilon^2/6$. Then for any $k$ dimensional shape $S$,
$
		\sum_i \sqrt{\dist(Ba_i^B,S)^2 + \apx_i^2} = (1 \pm 5\epsilon)\sum_i \dist(a_i,S).
$
	\label{thm:approximate-projections-to-approxiamte-distance}
\end{theorem}

The above theorem shows that we have to only compute approximate projections onto the subspace, which can be done in input sparsity time by using high probability subspace embeddings obtained from CountSketch matrices (see Section 2.3 of \cite{dw-sketching} and \cite{high_probability_embeddings}).
\section{Dimensionality Reduction for Sparse Inputs}\label{sec:mainsparse}

\subsection{Constructing an \texorpdfstring{$(O(1),\tilde{O}(k))$}{(sqrt(k),k)}-bicriteria Subspace Approximation}\label{sec:polyksubspaceapprox}
We first show how to obtain an $(O(1),\tilde{O}(k))$-bicriteria solution for $(k,1)$-subspace approximation. A key tool we use is a lopsided embedding defined as follows:

	\begin{definition}[Lopsided embedding]
		A matrix $S$ is a lopsided $\epsilon$-embedding for matrices $A$ and $B$ with respect to a matrix norm $\|\cdot\|$ and constraint set $\mathcal{C}$, if (i) for all matrices $X$ of the appropriate dimensions, $\|S(AX-B)\| \ge (1-\epsilon)\|AX-B\|$, and
			(ii) 
			for $B^* = AX^* - B$, we have
			$
				\|SB^*\|\leq(1+\epsilon)\|B^*\|,
			$
			where $X^* = \argmin_{X \in \mathcal{C}}\|AX - B\|$.
	\end{definition} 

Let $U_k \in \R^{n \times k}$ and $V_k^\mathsf{T} \in \R^{k \times d}$ be rank $k$ matrices such that $\|U_kV_k^\mathsf{T} - A\|_{1,2} = \textnormal{SubApx}_{k,1}(A)$. \citet{clarkson2015input} show that if $S$ is a lopsided $\epsilon$-embedding for matrices $(V_k,A^\mathsf{T})$ with respect to the norm $\|\cdot\|_h$, then $\min_{\text{rank-}k\ X}\|AS^\mathsf{T}X - A\|_{1,2} \le (1 + O(\epsilon))\textnormal{SubApx}_{k,1}(A)$. We show that a suitably scaled Gaussian random matrix $\bS$ with $\tilde{O}(k)$ rows is a lopsided $(1/4)$-embedding for matrices $(V_k,A^\mathsf{T})$ with probability $\ge 9/10$. Thus, we have that with probability $\ge 9/10$, 
\begin{equation*}
	\min_{\text{rank-}k\ X}\|A\bS^\mathsf{T}X - A\|_{1,2} \le (3/2)\textnormal{SubApx}_{k,1}(A).
\end{equation*}

We next prove that a row-sampling based $\ell_1$ subspace embedding for the column space of the matrix $A\bS^\mathsf{T}$ can be used to obtain a bicriteria solution to the subspace approximation problem. 

The following lemma summarizes the results discussed above. The results of the lemma are a significant improvement over Lemma~44 of \cite{clarkson2015input} and have simpler proofs that do not involve $\epsilon$-nets.
\begin{lemma}
	(i) If $\bS^\mathsf{T}$ is a random Gaussian matrix with $O(k)$ columns, then $\bS$ is a $1/4$-lopsided embedding for $(V_k, A^\mathsf{T})$ with respect to the $\| \cdot \|_{h}$ norm with probability $\ge 9/10$. Therefore, with probability $\ge 9/10$
	\begin{equation*}
		\min_{\text{rank-}k\ X}\|A\bS^\mathsf{T}X - A\|_{1,2} \le (3/2)\textnormal{SubApx}_{k,1}(A).
	\end{equation*}	
	(ii) If $\bL$ is a random matrix drawn from a distribution such that with probability $\ge 9/10$,
	$
		\alpha\|A\bS^\mathsf{T}y\|_{1} \le \|\bL A\bS^\mathsf{T}y\|_{1} \le \beta\|A\bS^\mathsf{T}y\|_1
	$
	for all vectors $y$ and if
	$
		\E_\bL[\|\bL M\|_{1,2}] = \|M\|_{1,2}
	$
	for any matrix $M$, then with probability $\ge 3/5$, all matrices $X$ of appropriate dimensions such that $\|\bL A\bS^\mathsf{T}X - \bL A\|_{1,2} \le 10\cdot \textnormal{SubApx}_{k,1}(A)$ satisfy
	$
		\|A\bS^\mathsf{T} X - A\|_{1,2} \le O(2+40/\alpha) \cdot \textnormal{SubApx}_{k,1}(A).
	$
	\label{lma:sketch-to-original}
\end{lemma}
Using the above lemma, we now have the following theorem which shows that Algorithm~\ref{alg:polykapprox} returns an $(O(1),\tilde{O}(k))$ approximation.

	\begin{algorithm}
		\caption{\textsc{PolyApprox}}
		\label{alg:polykapprox}
		\begin{algorithmic}
		    \STATE{{\bfseries Input:} $A\in\R^{n\times d}$, $B \in \R^{d \times c_1}$, $k \in \mathbb{Z}, \delta$}
			\STATE{{\bfseries Output:} $\hat{X} \in \R^{d \times c_2}$}
			\STATE $\text{cols} \gets O(k + 1/\delta^2)$
			\STATE $\bS^\mathsf{T} \gets \mathcal{N}(0,1)^{d \times \text{cols}}$
			\STATE $\bL \gets$ \textsc{LewisWeight}$(A(I-BB^\mathsf{T})\bS^\mathsf{T}, $1/2$)$ \cite{cohen-peng}
			\STATE $\hat{X} \gets $ Orthonormal Basis for rowspace($\bL A(I - BB^\mathsf{T})$)
			\STATE Repeat the above $O(\log(1/\delta))$ times and return the best $\hat{X}$ i.e., $\hat{X}$ minimizing $\|A(I-\hat{X}\hat{X}^\mathsf{T})\bG\|_{1,2}$ where $\bG$ is a Gaussian matrix with $O(\log(n))$ columns
			\end{algorithmic}
	\end{algorithm}
	\begin{theorem}
		Given any matrix $A \in \R^{n \times d}$ and a matrix $B \in \R^{d \times c_1}$ with $c_1 = \poly(k/\epsilon)$ orthonormal columns, Algorithm~\ref{alg:polykapprox} returns a matrix $\hat{X}$ with $\tilde{O}(k)$ orthonormal columns that with probability $1-\delta$ satisfies
		\begin{align*}
			&\|A(I-BB^\mathsf{T})(I - \hat{X}\hat{X}^\mathsf{T})\|_{1,2}\\
			&\le O(1) \cdot \textnormal{SubApx}_{k,1}(A(I-BB^\mathsf{T})),
		\end{align*}
		in time  $\tilde{O}((\nnz(A) + d\poly(k/\epsilon))\log(1/\delta))$.
		\label{thm:poly-k-approx}
	\end{theorem}

\subsection{Constructing a \texorpdfstring{$(1+\epsilon, \tilde{O}(k^{3}/\epsilon^2))$}{(1+eps,poly(k/eps))}-bicriteria Subspace Approximation}\label{sec:epssubspaceapprox}
Using the $(O(1),\tilde{O}(k))$-bicriteria subspace approximation solution found, we design a finer sampling process based on Theorem~45 of \cite{clarkson2015input} to further pick a subspace of dimension $\tilde{O}(k^{3}/\epsilon^2)$ that contains a $(1+\epsilon)$-approximate solution for subspace approximation of the matrix $A(I-BB^\mathsf{T})$. 

The following lemma states that given a subspace of cost at most $K\cdot \textnormal{SubApx}_{k,1}(A)$, that a sample of $\tilde{O}(K \cdot k^3/\epsilon^2)$ rows with probabilities chosen proportional to the distances of the rows of the matrix $A$ to the subspace, can be used to construct a subspace that is a $1+\epsilon$ approximation. 

		\begin{algorithm}
	\caption{\textsc{EpsApprox}}
	\label{alg:epsapprox}
	\begin{algorithmic}
	    	\STATE {\bfseries Input: } $A, B, \hat{X}, k, K, \epsilon,\delta > 0$.
	        \STATE {\bfseries Output: }$U \in \mathbb{R}^{d \times c}$ such that $U^\mathsf{T}B = 0$.
	        \STATE $t \gets O(\log(n/\delta))$, $\bG \gets \mathcal{N}(0,1/t)^{d \times t}$
	        \STATE $M \gets A(I - BB^\mathsf{T})(I-\hat{X}\hat{X}^\mathsf{T})\bG$
	        \STATE $p_i \gets \|M_{i*}\|_2/\|M\|_{1,2}$ for all $i \in [n]$
	        \STATE $s \gets \tilde{O}(K \cdot k^3/\epsilon^2 \cdot \log(1/\delta))$
	        \STATE $\bS \gets$ Multiset of $s$ independent samples drawn from distribution $p$
	        \STATE $U \gets$ Orthonormal basis for column space of the matrix $((I-BB^\mathsf{T})[\hat{X}\, (A_{\bS})^\mathsf{T}])$
	        \STATE {\bfseries Return} $U$
	\end{algorithmic}
	\end{algorithm}
	
\begin{lemma}
	Given a matrix $A \in \R^{n \times d}$ and a matrix $\hat{X} \in \R^{d \times c}$ that satisfies
	$$
		\|A(I-\hat{X}\hat{X}^\mathsf{T})\|_{1,2} \le K\cdot \textnormal{SubApx}_{k,1}(A),
	$$
	suppose we generate a matrix $\bS$ of $s = \tilde{O}((K/\alpha) \cdot k^3/\epsilon^2 \cdot \log(1/\delta))$ rows, each chosen independently to be the $i^{\textnormal{th}}$ standard basis vector with probability $p_i$. Here, $\sum_{i \in [n]} p_i = 1$ and for all $i \in [n]$
	$
		p_i \ge \alpha\frac{q_i}{\sum_i q_i},
	$
	where $q_i = \|A_{i*}(I-\hat{X}\hat{X}^\mathsf{T})\|_2$. 
	Let $U$ be an orthonormal basis for the rowspace of $[\hat{X}^\mathsf{T}\ ;\ \bS A]$. Then with probability $\ge 1 - \delta$,
	$$
		\|A(I-UU^\mathsf{T})\|_{1,2} \le (1+\epsilon)\textnormal{SubApx}_{k,1}(A).
	$$
	\label{lma:residual-sampling-main-lemma}
	\end{lemma}
	The proof of the above lemma is the same as that of the proof of Theorem~45 of \cite{clarkson2015input} with a minor change to account for the approximation error $\alpha$. Now the following theorem shows that Algorithm~\ref{alg:epsapprox} satisfies conditions of the previous lemma.
  	\begin{theorem}[Residual Sampling]\label{thm:consttoepsapprox}
	    Given matrix $A \in \mathbb{R}^{n \times d}$, matrices $B \in \mathbb{R}^{d\times c_1}$ and $\hat{X} \in \R^{d \times c_2}$ with orthonormal columns such that
$
	        \|A(I-BB^\mathsf{T})(I-\hat{X}\hat{X}^\mathsf{T})\|_{1,2} \le K\cdot \textnormal{SubApx}_{k,1}(A(I-BB^\mathsf{T})),
$
	    Algorithm~\ref{alg:epsapprox} returns a matrix $U$ having $c = \tilde{O}(c_2 + K \cdot k^3/\epsilon^2\cdot \log(1/\delta))$ orthonormal columns such that with probability $\ge 1 - \delta$,
	    \begin{align*}
	    	&\|A(I-BB^\mathsf{T})(I-UU^\mathsf{T})\|_{1,2} \\
	    	&\le (1+\epsilon)\textnormal{SubApx}_{k,1}(A(I-BB^\mathsf{T})).
	   \end{align*}
	       The algorithm runs
	    in time $\tilde{O}(\nnz(A) + d\poly(k/\epsilon))$. Moreover we also have that $U^\mathsf{T}B = 0$ i.e., the column spaces of $U$ and $B$ are orthogonal to each other.
	    \label{thm:algo-eps-approx-proof}
	\end{theorem}
	The proof of the theorem mainly involves showing that $\|M_{i*}\|_2$ is proportional to the residual $\|A_{i*}(I-BB^\mathsf{T})(I-\hat{X}\hat{X}^\mathsf{T})\|_2$. This is done by using the fact that if $\bG$ is a Gaussian matrix with $O(\log(1/\delta))$ rows, then $\|x^\mathsf{T}\bG\|_2 = (1/2,3/2)\|x^\mathsf{T}\|_2$ with probability $\ge 1 - \delta$. We then apply Lemma~\ref{lma:residual-sampling-main-lemma} to conclude that the solution computed by the algorithm is a bicriteria solution of cost at most $(1+\epsilon)\textnormal{SubApx}_{k,1}(A(I-BB^\mathsf{T}))$. Therefore, using the $(O(1), \tilde{O}(k))$ bicriteria solution obtained using Algorithm~\ref{alg:polykapprox}, we can obtain a $(1+\varepsilon, \tilde{O}(k^{3}/\varepsilon^2))$ bicriteria solution.
\section{Dimensionality Reduction}\label{sec:dimensionality-reduction}
  With an algorithm to construct a $(1+\epsilon, k^{3}/\epsilon^2)$ bicriteria solution from the previous section, we are now ready to construct a subspace that satisfies \eqref{eqn:extension-property}. Recall the crucial property for the subspace $S$ we need is that for all $k$-dimensional subspaces $W$,
  $
  	\|A(I-\P_S)\|_{1,2} - \|A(I-\P_{S + W})\|_{1,2} \le \epsilon^2\text{SubApx}_{k,1}(A).
  $
  To get such a subspace, we run Algorithms~\ref{alg:polykapprox} and \ref{alg:epsapprox} adaptively and then show that the union of all $1+\epsilon$ approximate bicriteria solutions satisfy the above property with parameter $O(\varepsilon)$. Thus, running the algorithm with parameter $\Theta(\epsilon^2)$ gives a subspace with the desired property.
  	\begin{algorithm}[t]
		\caption{\textsc{DimensionReduction}}
		\label{alg:dimensionreduction}
		\begin{algorithmic}
		    \STATE {\bfseries Input:} $A\in\R^{n\times d}$, $k, \epsilon>0$.
			\STATE {\bfseries Output:} $B \in \R^{d \times c}$ with orthonormal columns
			 \STATE $i^* \leftarrow $ uniformly random integer from $[10/\epsilon+1]$.
			 \STATE Initialize $B \gets []$
			\FOR {$i^*$ iterations}
			
			     \STATE $\hat{X} \leftarrow$ \textsc{PolyApprox}$(A,B,k,\epsilon/100)$.
				 \STATE $U \leftarrow$ \textsc{EpsApprox}$(A,B,\hat{X},k, \tilde{O}(\sqrt{k}), \epsilon,\epsilon/100)$.
				 \STATE $B \leftarrow [B\, |\, U ]$.
			\ENDFOR
			\STATE {\bfseries Return} $B$.
		\end{algorithmic}
	\end{algorithm} 
	\begin{algorithm}[t]
	\caption{\textsc{CompleteDimReduce}}
	 \label{alg:complete-dim-reduce}
	 \begin{algorithmic}
	    \STATE {\bfseries Input:} $A \in \mathbb{R}^{n \times d}$, $k \in \mathbb{Z}$, $\epsilon>0$.
	    \STATE {\bfseries Output:} $\text{Apx} \in \R^{n \times (c+1)}$
	    \STATE Let $B= $ \textsc{DimensionReduction}$(A, k, \Theta(\epsilon^2))$.\;
	    \STATE $t = O(\log(n))$\;
	    \STATE Compute $(\bS_j B, \bS_j A^\mathsf{T})$ for $j \in [t]$ where $\bS_j$ is an independent CountSketch matrix with $\poly(k/\epsilon)$ rows
	    \FOR{$i=1,\ldots,n$}
		\STATE Let $U_jD_j{V_{j}^\mathsf{T}} \gets \text{SVD}(\bS_j[B\, A_{i*}^\mathsf{T}])$ for all $j \in [t]$

		\FOR{$j \in [t]$}
			\STATE Check if for at least half $j' \ne j$, all singular values of $D_j{V_{j}^\mathsf{T}}{V_{j'}}(D_{j'}^\mathsf{T})^{-1}$ are in $[1-\Theta(\epsilon^2),1+\Theta(\epsilon^2)]$
			\STATE If the above check holds, set $x_i \gets (\bS_j B)^{\dagger}(\bS_j A_{i*}^\mathsf{T})$, $v_i \gets {\|(I - (\bS_j B)(\bS_j B)^\dagger)(\bS_jA_{i*}^\mathsf{T})\|_2}$ and go to next $i$
		\ENDFOR
		\ENDFOR
	\STATE {\bfseries Return} $B$ and $n \times (c+1)$ matrix $\textnormal{Apx}$ with $\textnormal{Apx}_{i*} = [x_i\, v_i]$
	 \end{algorithmic}
	\end{algorithm}

\begin{theorem}
    Given an $n \times d$ matrix $A$, $k \in \mathbb{Z}$, and an accuracy parameter $\epsilon > 0$, Algorithm~\ref{alg:complete-dim-reduce} returns a matrix $B$ with $\tilde{O}(k^{3}/\epsilon^6)$ orthonormal columns and a matrix $\text{Apx} = [X\, v]$ such that, with probability $\ge 9/10$, for any $k$ dimensional shape $S$, 
$
        \sum_i \sqrt{\dist(BX_{i*}^\mathsf{T}, S)^2 + v_i^2} = (1 \pm\varepsilon)\sum_i \dist(A_i, S).
$
    The algorithm runs in time $O(\nnz{(A)}/\epsilon^2 + (n+d)\poly(k/\epsilon))$.
\end{theorem}
 Let $B_i$ be the value of the matrix $B$ after $i$ iterations in Algorithm~\ref{alg:dimensionreduction}. The proof of the above theorem first shows that Algorithm~\ref{alg:dimensionreduction} outputs a subspace $B$ satisfying \eqref{eqn:extension-property}. This is done by showing that for at least a constant fraction of $j \in [10/\epsilon+1]$, the terms $\|A(I-B_{j}B_{j}^\mathsf{T})\|_{1,2}$ and $\|A(I-B_{j+1}B_{j + 1}^\mathsf{T})\|_{1,2}$ are close. This further means that the rows of the matrix $A(I-B_{j}B_{j}^\mathsf{T})$ cannot be projected onto any $k$ dimensional subspace $W$ to make $\|A(I-B_{j}B_{j}^\mathsf{T})(I-WW^\mathsf{T})\|_{1,2}$ substantially smaller than $\|A(I-B_{j}B_{j}^\mathsf{T} )\|_{1,2}$. Thus, we can show that with constant probability, for $i^*$ chosen randomly by Algorithm~\ref{alg:dimensionreduction}, the subspace $\text{colspan}(B_{i^*})$ satisfies \eqref{eqn:extension-property}. 
 
 Then, the proof uses the fact that for every $i \in [n]$, the algorithm finds a matrix $\bS_j$ that is a $\Theta(\epsilon^2)$ subspace embedding for $[B\, A_{i*}^\mathsf{T} ]$. This is shown to be true in \cite{high_probability_embeddings}. Now, if $\bS_j$ is a subspace embedding, it can be shown that the vector $x_i$ and value $v_i$ satisfy the conditions of Theorem~\ref{thm:approximate-projections-to-approxiamte-distance}, thus proving the above theorem.

\section{Linear Time Algorithm for Dense Matrices}\label{sec:fastcomputation}
We see from Algorithm~\ref{alg:complete-dim-reduce} that, after computing a subspace that satisfies \eqref{eqn:extension-property}, we can compute approximate projections and approximate distances to the subspace in time $\tilde{O}(nd + (n+d)\poly(k/\epsilon))$. We now show that the subspace can also be found in $\tilde{O}(nd + (n+d)\poly(k/\epsilon))$ time, thereby giving a near linear time algorithm for dimensionality reduction for dense matrices.
\subsection{Computing an \texorpdfstring{$(O(1),\poly(k))$}{(poly(k),poly(k))} approximation}

Consider constructing an $\ell_1$ subspace embedding for the matrix $A(I-BB^\mathsf{T})\bS^\mathsf{T}$ in Algorithm~\ref{alg:polykapprox}. The algorithm uses Lewis weights to sample a matrix that is an $O(1)$ $\ell_1$ subspace embedding for $A(I-BB^\mathsf{T})\bS^\mathsf{T}$ with high probability. We instead use the following theorem to compute an $\ell_1$ subspace embedding which is more amenable for giving fast algorithms for dense matrices.
\begin{theorem}[Section~3.1 of \citet{dw-sketching}]
	Given a matrix $A \in \R^{n \times d}$, let $\bL \in \R^{r \times n}$  be a random matrix that is an $(\alpha,\beta)$ $\ell_1$ subspace embedding for the matrix $A$.
	Let $\bL A = QR$ be the QR decomposition of the matrix  $\bL A$. Let $\ell_i = \|A_{i*}R^{-1}\|_1$ for $i \in [n]$. 
	If we generate a matrix $\bL'$ with $N = O((d^2\sqrt{r}/\gamma\epsilon^2)(\beta/\alpha)\log(1/\delta\epsilon))$ rows, each chosen independently as the $i^\textnormal{th}$ standard basis vector, times $1/(Np_i)$ with probability $p_i$, where $p_i \ge \gamma(\ell_i/\sum_{i'}\ell_{i'})$ for all $i \in [n]$, then the matrix $\bL$ satisfies with probability $1-\delta$, for all vectors $x$, 
$
		(1-\epsilon)\|Ax\|_1 \le \|\bL'Ax\|_1 \le (1+\epsilon)\|Ax\|_1.
$
	\label{thm:subspace-embedding-to-leverage-score-sampling}
\end{theorem}

Therefore, given an $(\alpha,\beta)$ $\ell_1$ subspace embedding with $r$ rows for the matrix $A(I-BB^\mathsf{T})\bS^\mathsf{T} \in \R^{n \times O(k)}$, we can compute a $(1 \pm \epsilon)$ $\ell_1$ embedding with $N = O((k^2\sqrt{r}/\gamma\epsilon^2)(\beta/\alpha)\log(1/\epsilon))$ rows. Using the $\ell_1$ subspace embedding of Theorem~1.3 of \citet{wang-woodruff}, we have $r,(\beta/\alpha) = O(k\log(k))$.

\begin{theorem}
	Given $A \in \R^{n \times d}$, $B \in \R^{d \times c_1}$, $k \in \mathbb{Z}$ and $\delta$, there exists an algorithm that returns $\hat{X}$ with $\tilde{O}(k^{3.5})$ orthonormal columns that with probability $1-\delta$ satisfies
		\begin{align*}
			&\|A(I-BB^\mathsf{T})(I - \hat{X}\hat{X}^\mathsf{T})\|_{1,2}\\
			&\le O(1) \cdot \textnormal{SubApx}_{k,1}(A(I-BB^\mathsf{T})).
		\end{align*}
Given that the matrices $\bC_1A_{I_j}$ for all $j \in [b]$ and $\bW A$, where $\bC_1$ is a Cauchy matrix with $O(\log(n\poly(k/\epsilon)))$ rows and $\bW$ is the subspace embedding of \citet{wang-woodruff} for $O(k)$ dimensional spaces, are precomputed for each of $O(\log(1/\delta))$ trials, the algorithm can be implemented in time $\tilde{O}(((nd/b)\cdot k^{3.5} + d \cdot \poly(k/\epsilon))\log(1/\delta))$.
\end{theorem}

We first partition rows of the matrix $A$ into $[b]$ sets denoted $I_1,\ldots, I_b$. To prove the above theorem, we use the following fact: if $\bC$ is a Cauchy matrix with $O(\log(n/\delta)/\epsilon^2)$ rows, then for any vector $x$ of $n$ dimensions, $\text{median}(\text{abs}(\bC x)) = (1 \pm \epsilon)\|x\|_1$ with probability $\ge 1 - \delta$. This fact lets us compute an approximation to the sum of leverage scores of the rows that lie in $I_j$ quickly without computing individual leverage scores. Using these approximations, we can sample $r$ rows by computing the leverage scores of all rows in just $r$ blocks instead of computing the leverage scores of all the rows of the matrix, which gives the cost saving as described in the theorem.

\subsection{Computing a \texorpdfstring{$(1+\epsilon,\poly(k/\epsilon))$}{(1+eps,poly(k/eps))} approximation}
\begin{theorem}
	Given a matrix $A \in \R^{n \times d}$, orthonormal matrices $B$ and $\hat{X}$ such that
	$
		\|A(I-BB^\mathsf{T})(I-\hat{X}\hat{X}^\mathsf{T})\|_{1,2} \le K \cdot \textnormal{SubApx}_{k,1}(A(I-BB^\mathsf{T})),
	$
	and parameters $k$, $\epsilon$, and $\delta$, there exists an algorithm that outputs a matrix $U$ with $\poly(K \cdot k/\epsilon)$ orthonormal columns such that with probability $\ge 1 - \delta$,
	$
		\|A(I-BB^\mathsf{T})(I-UU^\mathsf{T})\|_{1,2} \le (1+\epsilon)\textnormal{SubApx}_{k,1}(A(I-BB^\mathsf{T})).
	$
	Given that $\bC_1 A_{I_j}$ is precomputed for all $j \in [b]$, where $\bC_1$ is a Cauchy Matrix with $O(\log(n\poly(k/\epsilon)))$ rows, the algorithm runs in time $\tilde{O}((nd/b) \cdot (K \cdot k^{3}/\epsilon^2 \log(1/\delta)) + d\poly(k/\epsilon))$.
\end{theorem}
Note that Algorithm~\ref{alg:epsapprox} samples $O(K \cdot \poly(k/\epsilon))$ rows using probabilities proportional to the residuals of the rows with respect to the $O(1)$ approximation. We again divide the rows of the matrix $A(I-BB^\mathsf{T})(I-\hat{X}\hat{X}^\mathsf{T})$  into $b$ parts denoted by $I_1,\ldots, I_b$. We use the following fact from \cite{plain-vershynin-one-bit-compressive-sensing}: If $\bG$ is a Gaussian matrix with $m$ columns, then with high probability $(1/m)\|x^\mathsf{T}\bG\|_1 \approx (\sqrt{2/\pi})\|x\|_2$. Thus $\|A_{I_j*}(I-BB^\mathsf{T})(I-\hat{X}\hat{X}^\mathsf{T})\|_{1,2}$ can be approximated by scaling $\|A_{I_j*}(I-BB^\mathsf{T})(I-\hat{X}\hat{X}^\mathsf{T})\bG\|_{1,1}$, i.e., the sum of absolute values of entries of the matrix $A_{I_j*}(I-BB^\mathsf{T})(I-\hat{X}\hat{X}^\mathsf{T})\bG$. We show that these approximations can be computed quickly given the precomputed matrices as required in the theorem statement. Given the approximations for the sum of residuals of rows in each $I_j$, we only have to compute the residuals of $(n/b)c$ rows to sample $c$ rows from the distribution of residuals.

Replacing Algorithms~\ref{alg:polykapprox}~and~\ref{alg:epsapprox} by algorithms given by the above theorems lets us compute a subspace satisfying \eqref{eqn:extension-property} in time $\tilde{O}((nd/b)k^{3.5}/\epsilon^6 + (n+d)\poly(k/\epsilon) + T)$, where $T$ is the time required to compute the precomputed matrices required by the algorithms. By choosing $b = k^{3.5}/\epsilon^6$, we obtain that a subspace can be computed in time $\tilde{O}(nd + (n+d)\poly(k/\epsilon) + T)$. Notice that the above theorems only require at most $\poly(k/\epsilon)$ products of the form $MA$ for matrices $M$ of at most $\poly(k/\epsilon)$ rows. These products can be obtained by computing $\mathcal{M}A$, where $\mathcal{M}$ is formed by stacking all the matrices $M$. Now $\mathcal{M}$ only has $\poly(k/\epsilon)$ rows and the product $\mathcal{M}A$ can be computed in time $\tilde{O}(nd)$ by using the fast rectangular matrix multiplication algorithm of \citet{coppersmith}. Thus, a subspace satisfying \eqref{eqn:extension-property} can be computed in time $\tilde{O}(nd + (n+d)\poly(k/\epsilon))$.
\section{Experiments}
We perform experiments to empirically verify that we can attain a non-trivial amount of data reduction while still being able to compute an approximate sum of distances to a $k$-dimensional shape. In our experiments, we set $n = 10000$ and $k = 5$. We use various subspaces to compute an approximation to the sum of distances to a $k$ center set.
\subsection{Synthetic Data}
We set $d = 10000$ and we choose a set $C$ of $5$ centers in $\R^d$ randomly. For each center, we add Cauchy noise to generate $2000$ independent samples and hence obtain a dataset $A$ of size $10000$. We run our algorithm with a target dimension of $100$ and store the subspaces at intermediate steps of the loop in Algorithm~\ref{alg:dimensionreduction}. We compare the approximation computed for $\sum_i \dist(A_{i*}, C)$ using the subspace computed by Algorithm~\ref{alg:dimensionreduction}, a random subspace, and the top singular value subspace of the same dimensions. See Figure~\ref{fig:synthetic-data} for a plot. We note that for all dimensions, the subspace output by Algorithm~\ref{alg:dimensionreduction} does better than a random subspace to approximate the sum of distances and particularly at lower dimensions, the subspace output by Algorithm~\ref{alg:dimensionreduction} does better than the top singular value subspace of the same dimension.
\begin{figure}
    \centering
    \includegraphics[width=\linewidth]{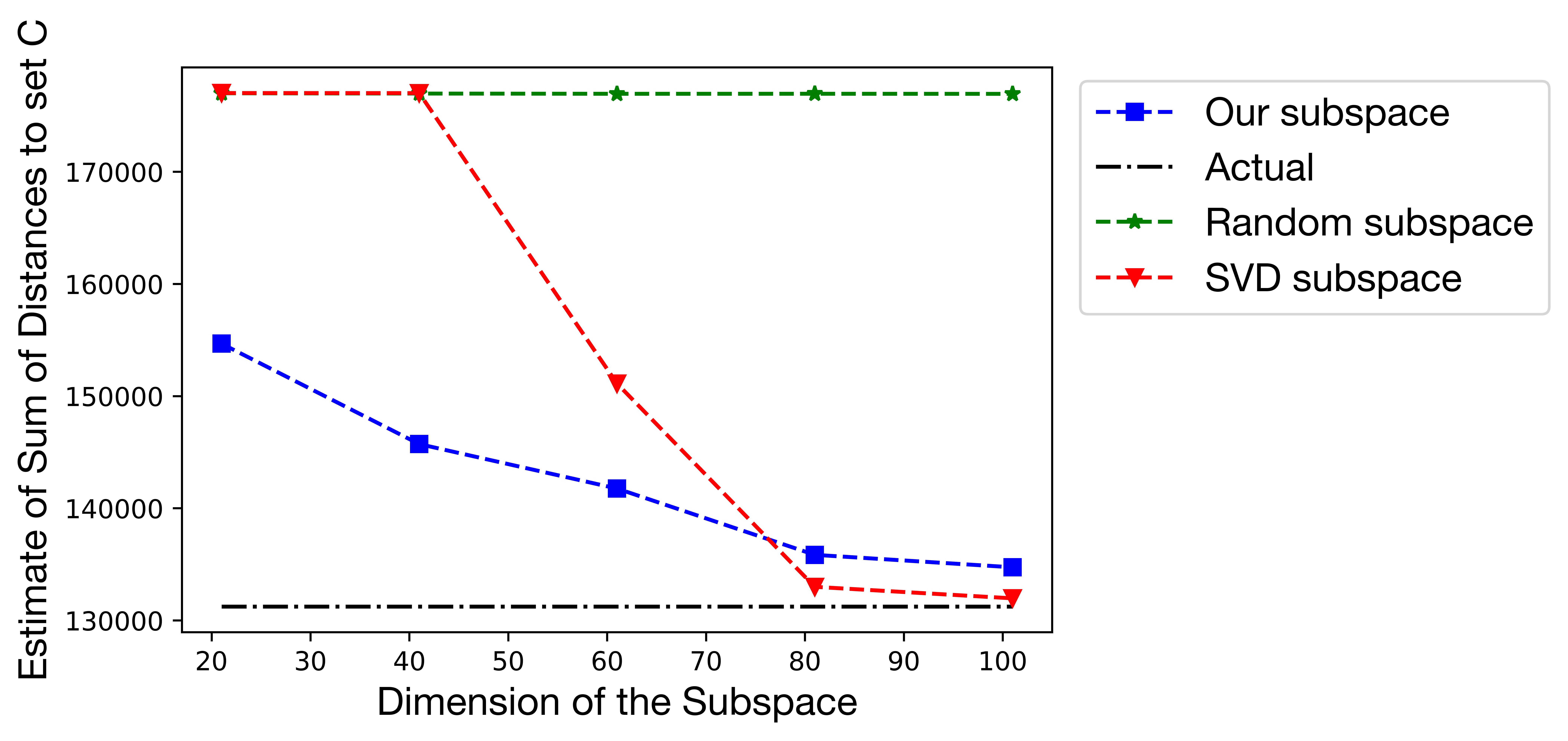}
    \caption{Comparison of subspaces output by Algorithm~\ref{alg:dimensionreduction} on \texttt{Synthetic} dataset with Random and Singular Value Subspaces}
    \label{fig:synthetic-data}
\end{figure}
\begin{figure}[h]
    \centering
    \includegraphics[width=\linewidth]{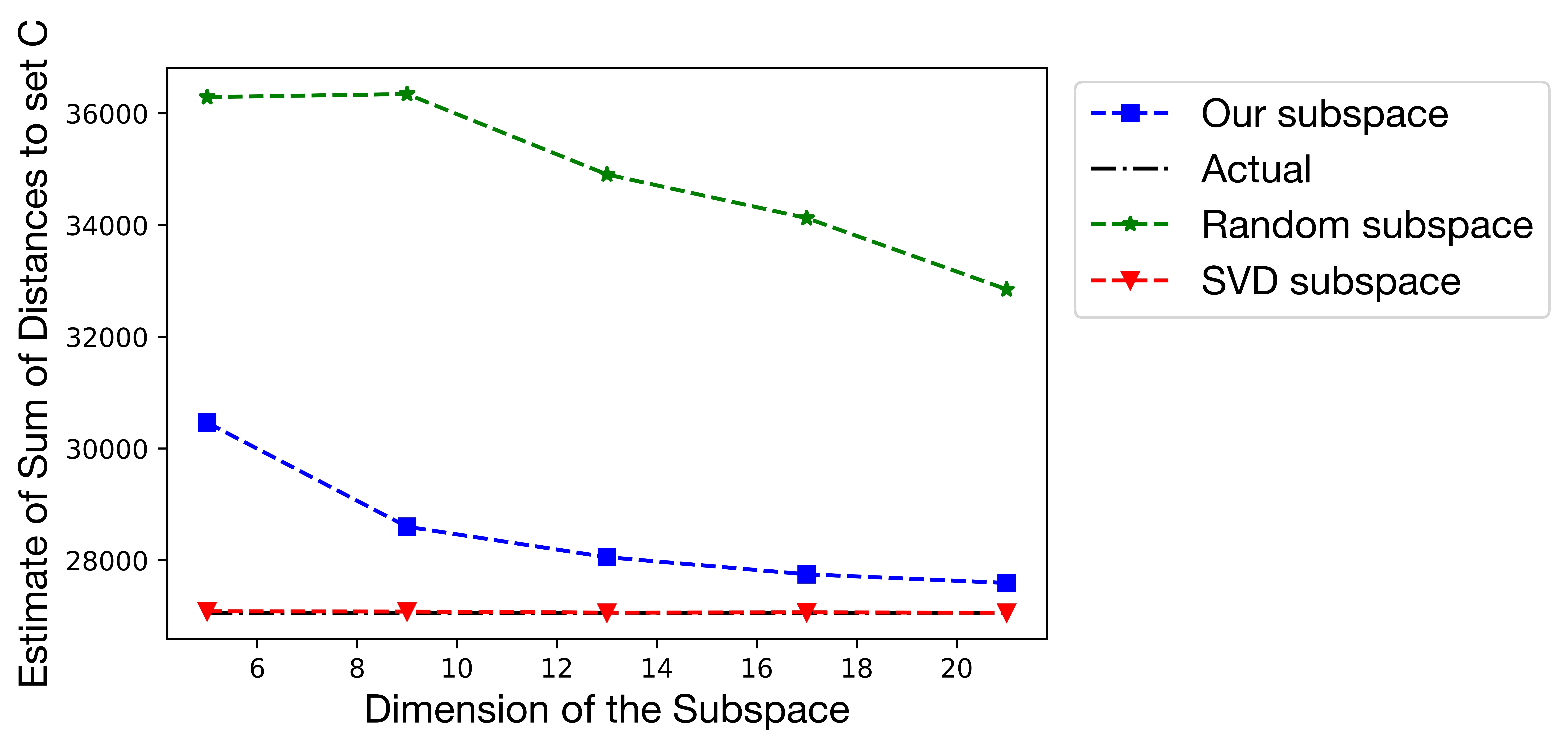}
    \caption{Comparison of subspaces output by Algorithm~\ref{alg:dimensionreduction} on \texttt{CoverType} dataset with random and singular value subspaces}
    \label{fig:covtype-data}
\end{figure}
\subsection{Real-World Data}
We run our dimensionality reduction algorithm on a randomly chosen subset $A$ of size $10000$ of the  \texttt{CoverType} dataset \cite{covtype}. We compute a $k$-means solution $C$ on the dataset and then evaluate the sum of distances to the center set $C$. Similar to the case of synthetic data, we compare the approximate sum of distances to $C$ computed using the subspace output by Algorithm~\ref{alg:dimensionreduction}, a random subspace, and the top singular value subspace. See Figure~\ref{fig:covtype-data} for a plot. We note that, again, the subspace output by our algorithm performs better than the random subspace at all dimensions. But note that the singular value subspace approximates the sum of distances to the set $C$  better than the subspace output by our algorithm. This occurs due to the fact that the data is inherently low dimensional and therefore if $P$ is the top singular value subspace, then $\P_Pa_i \approx a_i$ and $\dist(a_i, P) \approx 0$ and therefore, $\sqrt{\dist(a_i,P)^2 + \dist(\P_Pa_i, C)^2} \approx \dist(a_i, C)$. Thus, if the matrix $A$ can be approximated well with a low dimensional matrix, then we can instead use the top singular value subspace to reduce the dimension of the data and still be able to compute an approximate sum of distances to $k$ dimensional shapes, although we do not have any theoretical bounds for singular value subspaces.
\section*{Code}
An implementation of our Algorithm~\ref{alg:dimensionreduction} and code for our experiments is available \href{https://gitlab.com/praneeth10/dimensionality-reduction-for-sum-of-distances}{here}\footnote{\url{https://gitlab.com/praneeth10/dimensionality-reduction-for-sum-of-distances}}.
\section*{Acknowledgements}
The authors would like to thank support from the National Institute of Health (NIH) grant 5R01 HG 10798-2, Office of Naval Research (ONR) grant N00014-18-1-256, and a Simons Investigator Award.
\bibliographystyle{plainnat}
\bibliography{main.bib}
\onecolumn
\appendix
\section{Missing proofs from Section~\ref{sec:sum-of-distances}}
\begin{theorem}[{A version of Lemma~14 of \cite{sohler2018strong}}]
	Let $P$ be an $r$ dimensional subspace of $\R^d$ such that
	\begin{equation*}
		\sum_i \dist(a_i,P) - \sum_i \dist(a_i,\spn(P \cup H)) \le \frac{\epsilon^2}{80}\textnormal{SubApx}_{k,1}(A)
	\end{equation*}
	for all $k$-dimensional subspaces $H$. Let $B \in \R^{d \times r}$ be an orthonormal basis for the subspace $P$. For each $a_i$, let $a_i^B \in \R^r$ be such that
	$
		\dist(a_i, Ba_i^B) \le (1+\epsilon_c)\dist(a_i,P)
	$
	and let $(1-\epsilon_c)\dist(a_i,P) \le \apx_i \le (1+\epsilon_c)\dist(a_i,P)$ for $\epsilon_c = \epsilon^2/6$. Then for any $k$ dimensional shape $S$,
	\begin{equation*}
		\sum_i \sqrt{\dist(Ba_i^B,S)^2 + \apx_i^2} = (1 \pm 5\epsilon)\sum_i \dist(a_i,S)
	\end{equation*}
\end{theorem}
\begin{proof}
We have by the Pythagorean theorem that $\dist(Ba_i^B,a_i)^2 = \dist(Ba_i^B,\P_Pa_i)^2 + \dist(a_i,P)^2 \le (1+3\epsilon_c)\dist(a_i,P)^2$ which implies that $\dist(Ba_i^B,\P_Pa_i)^2 \le (3\epsilon_{c})\dist(a_i,P)^2$.

Given a shape $S$, we partition $[n]$ into two sets \emph{small} and \emph{large}. We say $i \in [n]$ is \emph{small} if
$
    \dist(\P_Pa_i, S) \le \dist(\P_Pa_i, Ba_i^B).
$
In that case, $\dist(Ba_i^B,S)^2 \le 4\dist(\P_Pa_i, Ba_i^B)^2 \le 12\epsilon_c\dist(a_i,P)^2$ by the triangle inequality and $\sqrt{\dist(Ba_i^B,S)^2 + \apx_i^2} \le \sqrt{1 + 15\epsilon_c}\dist(a_i,P) \le \sqrt{1+15\epsilon_c}\sqrt{\dist(a_i,P)^2 + \dist(\P_Pa_i, S)^2}$. Similarly, $\sqrt{\dist(Ba_i^B,S)^2 + \apx_i^2} \ge \apx_i \ge (1-\epsilon_c)\dist(a_i,P) \ge (1-4\epsilon_c)\sqrt{\dist(\P_Pa_i,S)^2 + \dist(a_i, P)^2}$ by using the fact that $\dist(\P_Pa_i, S)^2 \le 3\epsilon_c\dist(a_i,P)^2$.

We say that any $i \in [n]$ that is not \emph{small}, is \emph{large}. By the triangle inequality, we obtain that
\begin{equation}
    \dist(\P_Pa_i, S) - \dist(\P_Pa_i, Ba_i^B) \le \dist(Ba_i^B,S) \le \dist(\P_Pa_i, S) + \dist(Ba_i^B, \P_Pa_i).
\end{equation}
As $i$ is \emph{large}, $\dist(\P_Pa_i, S) - \dist(\P_Pa_i, Ba_i^B) > 0$ and therefore by the AM-GM inequality, we obtain that
	\begin{equation*}
		\dist(Ba_i^B,S)^2 = (1 \pm \epsilon)\dist(\P_Pa_i,S)^2 + \left(1 \pm \frac{1}{\epsilon}\right)\dist(Ba_i^B,\P_Pa_i)^2.
	\end{equation*}
 Thus,
	$
		\dist(Ba_i^B,S)^2 \le (1+\epsilon)\dist(\P_Pa_i,S)^2 + ({2}/{\epsilon})(3\epsilon_c)\dist(a_i,P)^2
	$
	and
	$
		\dist(Ba_i^B,S)^2 \ge (1 - \epsilon)\dist(\P_Pa_i,S)^2 - ({1}/{\epsilon})(3\epsilon_c)\dist(a_i,P)^2.
	$
	Letting $\epsilon_c = \epsilon^2/6$, we finally have
	\begin{equation*}
		\dist(Ba_i^B,S)^2 + \apx_i^2 \le (1+\epsilon)\dist(\P_Pa_i,S)^2 + (1 + 2\epsilon)\dist(a_i,P)^2
	\end{equation*}
	and
	\begin{equation*}
		\dist(Ba_i^B,S)^2 + \apx_i^2 \ge (1-\epsilon)\dist(\P_Pa_i,S)^2  + (1 - 3\epsilon)\dist(a_i,P)^2.
	\end{equation*}
	Therefore, by combining both \emph{small} and \emph{large} indices,
	\begin{equation*}
		\sum_i \sqrt{\dist(Ba_i^B,S)^2 + \apx_i^2} \le \sqrt{1+O(\epsilon)}\sum_i \sqrt{\dist(\P_Pa_i,S)^2 +\dist(a_i,P)^2}
	\end{equation*}
	and
	\begin{equation*}
		\sum_i \sqrt{\dist(Ba_i^B,S)^2 + \apx_i^2} \ge \sqrt{1-O(\epsilon)}\sum_i \sqrt{\dist(\P_Pa_i,S)^2 +\dist(a_i,P)^2}.
	\end{equation*}
	The theorem now follows from Theorem~8 of \cite{sohler2018strong}.
\end{proof}
\section{Missing Proofs from Section~\ref{sec:mainsparse}}
\subsection{Lopsided Embeddings and Gaussian Matrices}
Recall $\|\cdot\|_h$ is defined  as $\|A\|_h = \sum_j \|A_{*j}\|_2$. Note that $\|A\|_h = \|A^\mathsf{T}\|_{1,2}$ for all matrices $A$. The following lemma shows that lopsided-$\epsilon$ embeddings for certain matrices w.r.t.  the norm $\|\cdot\|_h$ imply a dimension reduction for $\|\cdot\|_{1,2}$ subspace approximation.
\begin{lemma}
	Given a matrix $A \in \R^{n \times d}$ and a parameter $k \in \mathbb{Z}_{> 0}$, let $U_k \in \R^{n \times k}$ and $V_k^\mathsf{T} \in \R^{k \times d}$ be matrices such that
	\begin{equation*}
		\|U_kV_k^\mathsf{T} - A\|_{1,2} = \min_{\text{rank-}k\ X}\|A(I-X)\|_{1,2}.
	\end{equation*}
	If $S$ is a lopsided $\epsilon$-embedding for $(V_k, A^\mathsf{T})$ with respect to the norm  $\|\cdot\|_h$, then
	\begin{equation*}
		\min_{\text{rank-}k\ X}\|AS^\mathsf{T}X - A\|_{1,2} \le (1+O(\epsilon))\min_{\text{rank-k}\ X}\|A(I-X)\|_{1,2}.
	\end{equation*}
	\label{lma:lopsided-to-epsilon-approximation}
\end{lemma}
\begin{proof}
	Note that $\|V_kU_k^\mathsf{T} - A^\mathsf{T}\|_h = \min_{Y}\|V_kY^\mathsf{T} - A^\mathsf{T}\|_h$. By definition of a lopsided embedding, we have the following for any matrix $Y$:
	\begin{equation*}
		\|YV_k^\mathsf{T}S^\mathsf{T} - AS^\mathsf{T}\|_{1,2} = \|SV_kY^\mathsf{T} - SA^\mathsf{T}\|_h \ge (1-\epsilon)\|V_kY^\mathsf{T} - A^\mathsf{T}\|_h = (1-\epsilon)\|YV_k^\mathsf{T} - A\|_{1,2}
	\end{equation*}
	and also that
	\begin{equation*}
		\|U_kV_k^\mathsf{T}S^\mathsf{T} - AS^\mathsf{T}\|_{1,2} = \|SV_kU_k^\mathsf{T} - SA^\mathsf{T}\|_h \le (1+\epsilon)\|V_kU_k^\mathsf{T} - A^\mathsf{T}\|_h = (1+\epsilon)\|U_kV_k^\mathsf{T} - A\|_{1,2}.
	\end{equation*}
	Using these guarantees we now show that the column span of the matrix $AS^\mathsf{T}$ contains a good solution to the subspace approximation problem. First consider the minimization problem
	\begin{equation*}
		\min_{Y}\|YV_k^\mathsf{T} - A\|_{1,2}.
	\end{equation*}
	Clearly, $U_k$ is the optimal solution to the problem. Now consider the optimal solution $\tilde{Y}$ to the sketched version of the above problem
	\begin{align*}
		\tilde{Y} = \argmin_{Y}\|YV_k^\mathsf{T}S^\mathsf{T} - AS^\mathsf{T}\|_{1,2}.
	\end{align*}
	We can see that $\tilde{Y} = (AS^\mathsf{T})(V_k^\mathsf{T}S^\mathsf{T})^+$. Now
	\begin{equation*}
		\|\tilde{Y}V_k^\mathsf{T} - A\|_{1,2} \le \frac{1}{1-\epsilon}\|\tilde{Y}V_k^\mathsf{T}S^\mathsf{T} - AS^\mathsf{T}\|_{1,2} \le \frac{1}{1-\epsilon}\|U_KV_k^\mathsf{T}S^\mathsf{T} - AS^\mathsf{T}\| \le \frac{1+\epsilon}{1-\epsilon}\|U_kV_k^\mathsf{T} - A\|_{1,2}.
	\end{equation*}
	Therefore, 
	\begin{equation*}
		\min_{\text{rank-k}\ X}\|AS^\mathsf{T}X - A\|_{1,2} \le \|AS^\mathsf{T}(V_k^\mathsf{T}S^\mathsf{T})^+(V_k^\mathsf{T}) - A\|_{1,2} \le \frac{1+\epsilon}{1-\epsilon}\|U_kV_k^\mathsf{T} - A\|_{1,2} \le (1+3\epsilon)\min_{\text{rank-}k\ X}\|A(I-X)\|_{1,2}.
	\end{equation*}
	Thus, if the number of rows of $S$ is less than $d$, we obtain a dimension reduction for $\|\cdot\|_{1,2}$ subspace approximation.
\end{proof}
\citet{clarkson2015input} give the following sufficient conditions for a distribution of matrices to be an $\epsilon$-lopsided embedding for $(A,B)$. For the sake of completeness we reproduce their proof here.
\begin{lemma}[Sufficient Conditions]
	Given matrices $(A,B)$, let $\bS$ be a matrix drawn from a distribution such that
	\begin{enumerate}
		\item the matrix $\bS$ is a subspace $\epsilon$-contraction for $A$ with respect to $\|\cdot\|_2$, i.e., simultaneously for all vectors $x$
		\begin{equation*}
			\|\bS Ax\|_2 \ge (1-\epsilon)\|Ax\|_2
		\end{equation*}
		with probability $1-\delta/3$,
		\item for all $i \in [d']$, with probability at least $1-\delta\epsilon^2/3$ the matrix $\bS$ is a subspace $\epsilon^2$-contraction for $[A\ B_{*i}]$ with respect to $\|\cdot\|_2$, i.e., for all vectors $x$,
		\begin{equation*}
			\|\bS Ax - \bS B_{*i}\|_2 \ge (1-\epsilon^2)\|Ax-B_{*i}\|_2,
		\end{equation*}
		 and
		\item the matrix $\bS$ is an $\epsilon^2$-dilation for $B^*$ with respect to $\|\cdot\|_h$, i.e., $\|\bS B^*\|_h \le (1+\epsilon^2)\|B^*\|_h$ with probability $\ge 1 - \delta/3$.
	\end{enumerate}
	In the Condition 3 above, $B^* = AX^*-B$ where $X^* = \argmin_{X}\|AX-B\|_h$. With failure probability at most $\delta$, the matrix $\bS$ is an affine $6\epsilon$-contraction for $(A,B)$ with respect to $\|\cdot\|_h$, i.e., for all matrices $X$,
	\begin{equation*}
		\|\bS(AX-B)\|_h \ge (1- 6\epsilon)\|AX-B\|_h
	\end{equation*}
	and therefore a lopsided $6\epsilon$-embedding for $(A,B)$ with respect to $\|\cdot\|_h$.
	\label{lma:lopsided-construction}
\end{lemma}
Importantly, note that Condition 2 in the lemma is about the probability of $\bS$ being a subspace contraction for $[A\, B_{*i}]$ separately for each $i$ and \emph{not} the probability of $\bS$ being \emph{simultaneously} a subspace contraction for $[A\, B_{*i}]$ for all $i \in [d']$.
\begin{proof}
	Condition on the event that 1 and 3 hold. For $i \in [d']$, let $\bZ_i$ be an indicator random variable where $\bZ_i = 0$ if the matrix $\bS$ is a subspace $\epsilon^2$-contraction for $[A\ B_{*i}]$ and $\bZ_i = 1$ otherwise. From the properties of $\bS$, we have that $\Pr[\bZ_i = 1] \le \delta\epsilon^2/3$ for all $i$. If $\bZ_i = 1$, we call $i$ \emph{bad} and if $\bZ_i = 0$, we call $i$ \emph{good}.
	
	Consider an arbitrary matrix $X$. Say a \emph{bad} $i$ is \emph{large} if $\|(AX-B)_{*i}\|_2 \ge (1/\epsilon)(\|B_{*i}\|_2 + \|\bS B_{*i}\|_2)$, otherwise a \emph{bad} $i$ is \emph{small}. We have
	\begin{equation}
		\sum_{\text{\emph{small}}\ i} \|(AX-B)_{*i}\|_2 \le (1/\epsilon)\sum_{\text{\emph{small}}\ i}\|B_{*i}\|_2 + \|\bS B_{*i}\|_2 \le (1/\epsilon)\sum_{\text{\emph{bad}}\ i}\|B_{*i}\|_2 + \|\bS B_{*i}\|_2.
		\label{eqn:small-i-upperbound}
	\end{equation}
	Using condition $2$, we obtain that $\E[\sum_{\text{\emph{bad}}\ i}\|B^*_{*i}\|_2] \le (\delta\epsilon^2/3)\sum_{i}\|B_{*i}^*\|_2 \le (\delta\epsilon^2/3)\Delta^*$. By a Markov bound, we have that with probability $\ge 1 - \delta/3$, $\sum_{\text{\emph{bad}}\ i}\|B_{*i}^*\| \le \epsilon^2\Delta^*$. Assume that this event holds. Similarly,
	\begin{align*}
		\sum_{\text{\emph{bad}}\ i}\|\bS B_{*i}^*\|_2 &= \|\bS B^*\|_h - \sum_{\text{\emph{good}}\ i}\|\bS B^*_{*i}\|_2\\
		&\le (1+\epsilon^2)\Delta^* - (1-\epsilon^2)\sum_{\text{\emph{good}}\ i}\|B_{*i}^*\|_2\\
		&\le (1+\epsilon^2)\Delta^* - (1-\epsilon^2)(\Delta^* - \epsilon^2\Delta^*)\\
		&\le 3\epsilon^2\Delta^*.
	\end{align*}
	Thus, we can bound the RHS of \eqref{eqn:small-i-upperbound} and obtain
	\begin{equation*}
		\sum_{\text{\emph{small}}\ i} \|(AX-B)_{*i}\|_2 \le (1/\epsilon)(\epsilon^2\Delta^* + 3\epsilon^2\Delta^*) \le 4\epsilon\Delta^*.
	\end{equation*}
	Now we lower bound $\sum_{\text{\emph{bad}}\ i} \|\bS (AX-B)_{*i}\|_2$.
	\begin{align*}
		\sum_{\text{\emph{bad}}\ i}\|\bS (AX-B)_{*i}\|_2 &\ge \sum_{\text{\emph{large}}\ i}\|\bS(AX-B)_{*i}\|_2\\
		&\ge \sum_{\text{\emph{large}}\ i}\|\bS(AX-AX^*)_{*i}\|_2 - \|\bS B^*_{*i}\|_2\\
		&\ge \sum_{\text{\emph{large}}\ i}(1-\epsilon)\|(AX-AX^*)_{*i}\|_2- \|\bS B^*_{*i}\|_2\\
		&\ge \sum_{\text{\emph{large}}\ i}(1-\epsilon)\|(AX-B)_{*i}\|_2- (1-\epsilon)\|B_{*i}^*\|_2 - \|\bS B^*_{*i}\|_2\\
		&\ge \sum_{\text{\emph{large}}\ i}(1-\epsilon)\|(AX-B)_{*i}\|_2 - \epsilon\|(AX-B)_{*i}\|_2\\
		&\ge (1-2\epsilon)\sum_{\text{\emph{large}}\ i}\|(AX-B)_{*i}\|_2.
	\end{align*}
	In the above, we repeatedly used the triangle inequality for the $\|\cdot \|_2$ norm, and that $\bS$ is a subspace $\epsilon$-embedding for matrix $A$ and for large $i$, we upper bound $(1-\epsilon)\|B^*_{*i}\|_2 + \|\bS B^*_{*i}\|_2$ by  $\epsilon\|(AX-B)_{*i}\|_2$. We can finally lower bound $\|\bS(AX-B)\|_h$.
	\begin{align*}
		\|\bS(AX-B)\|_h &= \sum_{\text{\emph{good}}\ i}\|\bS(AX-B)_{*i}\|_2 + \sum_{\text{\emph{bad}}\ i}\|\bS(AX-B)_{*i}\|_2\\
		&\ge (1-\epsilon^2)\sum_{\text{\emph{good}}\ i}\|(AX-B)_{*i}\|_2 + (1-2\epsilon)\sum_{\text{\emph{large}}\ i}\|(AX-B)_{*i}\|_2\\
		&\ge (1-\epsilon^2)\sum_{\text{\emph{good}}\, i}\|(AX-B)_{*i}\|_2 + (1-2\epsilon)\sum_{\text{\emph{bad}}\, i}\|(AX-B)_{*i}\|_2\\
		&\quad - (1-2\epsilon)\sum_{\text{\emph{small}}\, i}\|(AX-B)_{*i}\|_2\\
		&\ge (1-2\epsilon)\|AX-B\|_h - (1-2\epsilon)4\epsilon\Delta^*\\
		&\ge (1-6\epsilon)\|AX-B\|_h.
	\end{align*}
	Thus, by a union bound, with failure probability $\le \delta$, $\bS$ is an affine $6\epsilon$-contraction for $(A,B)$ with respect to  $\|\cdot\|_h$.
\end{proof}
\begin{lemma}[Gaussian Matrices are Lopsided Embeddings]
	Given arbitrary matrices $A$ of rank $k$ and $B$ of \emph{any} rank, a Gaussian matrix $\bS$ with $\tilde{O}(k/\epsilon^4 + 1/\epsilon^4\delta^2)$ rows is an $\epsilon$-lopsided embedding for $(A,B)$ with probability $\ge 1 - \delta$.
	\label{lma:gaussian-lopsided-embedding}
\end{lemma}
\begin{proof}
	We now show that a Gaussian matrix, with small dimension equal to $\tilde{O}(k/\epsilon^4 + 1/\epsilon^4\delta^2)$, satisfies all of the sufficient conditions of Lemma~\ref{lma:lopsided-construction}. Clearly, a Gaussian matrix with $O((k + \log(1/\delta))/\epsilon^2)$ rows satisfies condition 1 and a Gaussian matrix with $O((k + \log(1/\delta\epsilon))/\epsilon^4)$ rows satisfies condition 2 \cite{dw-sketching}.

We now show that a Gaussian matrix with at least $O(1/\epsilon^4)$ rows satisfies
\begin{equation*}
	\E[(\|\bS y\|_2^2 - 1)^2] \le \epsilon^4
\end{equation*}
for any given unit vector $y$. If $\bS$ is a Gaussian matrix of $t$ rows with each entry drawn i.i.d. from $N(0,1/t)$, then the entries of $Sy$ are each drawn i.i.d. from $N(0,\|y\|_2^2/t) = N(0,1/t)$. Therefore, $\|\bS y\|_2^2 = \bY_1^2 + \ldots + \bY_{t}^2$, where $\bY_i \sim N(0,1/t)$, which gives
\begin{align*}
	\E[(\|\bS y\|_2^2 - 1)^2] &= \E[(\bY_1^2 + \ldots + \bY_t^2 - 1)^2]\\
	&= t\E[\bY_1^4] + 1 + 2\binom{t}{2}\E[\bY_1^2\bY_2^2] - 2t\E[\bY_1^2] = t\frac{3}{t^2} + 1 + 2\binom{t}{2}\frac{1}{t^2} - 2t\frac{1}{t}\\
	&= 2/t.
\end{align*}
Thus, with $t \ge 1/\epsilon^4$, we have that $\E[(\|\bS y\|_2^2 - 1)^2] \le \epsilon^4$. By Lemma 28 of \cite{clarkson2015input}, we obtain that $\E[\max(\|\bS y\|_2^4,1)] \le (1+\epsilon^2)^2 \le 1+3\epsilon^2$. Now, by Holder's inequality,
\begin{equation*}
	\E[\max(\|\bS y\|_2,1)] \le \E[\max(\|\bS y\|_2,1)^4]^{1/4} \le (1+3\epsilon^2)^{1/4} \le 1 + (3/4)\epsilon^2.
\end{equation*}
As $(\|\bS y\|_2 - 1)_+ = \max(\|\bS y\|_2,1)-1$, we obtain that $\E[(\|\bS y\|_2 - 1)_+] \le (3/4)\epsilon^2$, which implies by scaling that for an arbitrary vector $y$,
\begin{equation*}
	\E[(\|\bS y\|_2 - \|y\|_2)_+] \le (3/4)\epsilon^2\|y\|_2
\end{equation*}
which gives
\begin{align*}
	\E[(\|\bS B^*\|_h - \|B^*\|_h)_+] &\le (3/4)\epsilon^2\|B^*\|_h.
\end{align*}
By Markov's inequality, with probability $\ge 1 - \delta/3$, $(\|\bS B^*\|_h - \|B^*\|_h)_+ \le (9/4)(\epsilon^2/\delta)\|B^*\|_h$ and hence, 
with probability $\ge 1 - \delta/3$, $\|\bS B^*\|_h \le (1+(9/4)(\epsilon^2/\delta))\|B^*\|_h$. Thus, a Gaussian matrix with $m = O(1/\epsilon^4\delta^2)$ rows satisfies that with probability $\ge 1 - \delta/3$ that
\begin{equation*}
	\|\bS B^*\|_h \le (1+\epsilon^2)\|B^*\|_h.\qedhere
\end{equation*}
\end{proof}
\subsection{Utilizing Sampling based \texorpdfstring{$\ell_1$}{l-1} embeddings}
Let $A$ be a matrix that has $r$ columns. Suppose $\bL$ is a random matrix such that with probability $\ge 9/10$, simultaneously for all vectors $y$,
\begin{equation*}
	\alpha\|Ay\|_1 \le \|\bL Ay\|_1 \le \beta\|Ay\|_1.
\end{equation*}
Assume the above event holds. Let $X$ be an arbitrary matrix with $t$ columns. We have that for a suitably scaled Gaussian matrix $\bG$ with $\tilde{O}(t/\epsilon^2)$ columns, with probability $\ge 9/10$, simultaneously for all vectors $x \in \R^t$, $\|x^{\mathsf{T}}\bG\|_1 = (1 \pm \epsilon)\|x\|_2$ \cite{matouvsek2013lecture}. Thus there exists a matrix $M$ with $\tilde{O}(t/\epsilon^2)$ columns such that for all vectors $x \in \R^t$,
\begin{equation*}
	\|x^\mathsf{T}M\|_1 = (1 \pm \epsilon)\|x\|_2.
\end{equation*}
Therefore,
\begin{equation*}
	\frac{1}{1+\epsilon}\|AXM\|_{1,1} \le \|AX\|_{1,2} = \frac{1}{1-\epsilon}\|AXM\|_{1,1}
\end{equation*}
and
\begin{equation*}
	\frac{1}{1+\epsilon}\|\bL AXM\|_{1,1} \le \|\bL AX\|_{1,2} \le \frac{1}{1-\epsilon}\|\bL AXM\|_{1,1}
\end{equation*}
Now we upper bound $\|\bL AX\|_{1,2}$.
\begin{align*}
	\|\bL AX\|_{1,2} &\le \frac{1}{1-\epsilon}\|\bL AXM\|_{1,1} \le \frac{1}{1-\epsilon}\sum_{j}\|\bL A(XM)_{*j}\|_1\\
	&\le \frac{\beta}{1-\epsilon}\sum_j\|A(XM)_{*j}\|_1 = \frac{\beta}{1-\epsilon}\|AXM\|_{1,1} \le \beta \frac{1+\epsilon}{1-\epsilon}\|AX\|_{1,2}.
\end{align*}
We now lower bound $\|\bL AX\|_{1,2}$ similarly.
\begin{align*}
	\|\bL AX\|_{1,2} &\ge \frac{1}{1+\epsilon}\|\bL AXM\|_{1,1} = \frac{1}{1+\epsilon}\sum_j \|\bL A(XM)_{*j}\|_1\\
	&\ge \frac{\alpha}{1+\epsilon}\sum_j \|A(XM)_{*j}\|_1 = \frac{\alpha}{1+\epsilon}\|AXM\|_{1,1} \ge \alpha\frac{1-\epsilon}{1+\epsilon}\|AX\|_{1,2}.
\end{align*}
By picking appropriate $\epsilon$, we conclude that for any matrix $X$,
\begin{equation}
	\frac{\alpha}{2}\|AX\|_{1,2} \le \|\bL AX\|_{1,2} \le 2\beta \|AX\|_{1,2}.
	\label{eqn:l1-to-1-2}
\end{equation}

\begin{lemma}
	If $\bS^\mathsf{T}$ is a random Gaussian matrix with $O(k)$ columns such that with probability $\ge 9/10$,
	\begin{equation*}
		\min_{\text{rank-}k\ X}\|A\bS^\mathsf{T}X - A\|_{1,2} \le (3/2)\min_{\text{rank-}k\ X}\|AX - A\|_{1,2},
	\end{equation*}
	and if $\bL$ is a random matrix drawn from a distribution such that with probability $\ge 9/10$ over the draw of matrix $\bL$, 
	\begin{equation*}
		\alpha\|A\bS^\mathsf{T}y\|_{1} \le \|\bL A\bS^\mathsf{T}y\|_{1} \le \beta\|A\bS^\mathsf{T}y\|_1
	\end{equation*}
	for all vectors $y$ and
	\begin{equation*}
		\E_\bL[\|\bL M\|_{1,2}] = \|M\|_{1,2}
	\end{equation*}
	for any matrix $M$, then with probability $\ge 3/5$, all matrices $X$ such that $\|\bL A\bS^\mathsf{T}X - \bL A\|_{1,2} \le 10\cdot \textnormal{SubApx}_{k,1}(A)$ satisfy
	\begin{equation*}
		\|A\bS^\mathsf{T} X - A\|_{1,2} \le \left(2 + 40/\alpha\right)\textnormal{SubApx}_{k,1}(A).
	\end{equation*}
	\label{lma:sketch-to-actual}
\end{lemma}
\begin{proof}
Let $X_1 = \argmin_{\text{rank-}k\ X}\|A\bS^\mathsf{T}X - A\|_{1,2}$. With probability $\ge 9/10$, we have that $\|A\bS^\mathsf{T} X_1 - A\|_{1,2} \le (3/2)\text{SubApx}_{k,1}(A)$. By a Markov bound, we obtain that with probability $\ge 4/5$, $\|\bL A\bS^\mathsf{T} X_1 - \bL A\|_{1,2} \le 10\textnormal{SubApx}_{k,1}(A)$. Assume this event holds.
For any matrix $X$,
\begin{equation*}
	\|\bL A\bS^\mathsf{T}X - \bL A\|_{1,2} \ge \|\bL A\bS^\mathsf{T}X - \bL A\bS^\mathsf{T}X_1\|_{1,2} - \|\bL A\bS^\mathsf{T}X_1 - \bL A\|_{1,2}.
\end{equation*}
We have
\begin{equation*}
	\|\bL A\bS^\mathsf{T}X - \bL A\|_{1,2} \ge \|\bL A\bS^\mathsf{T}X - \bL A\bS^\mathsf{T}X_1\|_{1,2} - 10\cdot \textnormal{SubApx}_{k,1}(A).
\end{equation*}
From \eqref{eqn:l1-to-1-2}, we have
\begin{align*}
	\|\bL A\bS^\mathsf{T}X - \bL A\|_{1,2} &\ge \frac{\alpha}{2}\|A\bS^\mathsf{T}X - A\bS^\mathsf{T}X_1\|_{1,2} - 10\cdot \textnormal{SubApx}_{k,1}(A)\\
	&\ge \frac{\alpha}{2}\|A\bS^\mathsf{T}X - A\|_{1,2}-\frac{\alpha}{2}\|A\bS^\mathsf{T}X_1 - A\|_{1,2} - 10 \cdot \textnormal{SubApx}_{k,1}(A)\\
	&\ge \frac{\alpha}{2}\|A\bS^\mathsf{T}X - A\|_{1,2} - (3\alpha/4 + 10)\cdot \textnormal{SubApx}_{k,1}(A).
\end{align*}
Thus, for any matrix $X$ of rank $r$, if $\|A\bS^\mathsf{T} X - A\|_{1,2} > (2/\alpha)(20 + 3\alpha/4)\cdot \textnormal{SubApx}_{k,1}(A)$, then $\|\bL A\bS^\mathsf{T}X - \bL A\|_{1,2} > 10\cdot \textnormal{SubApx}_{k,1}(A)$. 
\end{proof}
\subsection{Main Theorem for constructing an \texorpdfstring{$(O(1),\tilde{O}(k))$}{O(sqrt(k)),O(k)}-bicriteria solution}
	\begin{theorem}
		Given any matrix $A \in \R^{n \times d}$ and a matrix $B \in \R^{d \times c_1}$ with orthonormal columns, Algorithm~\ref{alg:polykapprox} returns a matrix $\hat{X}$ with $\tilde{O}(k)$ orthonormal columns that with probability $1-\delta$ satisfies
		\begin{equation*}
			\|A(I-BB^\mathsf{T})(I - \hat{X}\hat{X}^\mathsf{T})\|_{1,2} \le O(1) \cdot \textnormal{SubApx}_{k,1}(A(I-BB^\mathsf{T})),
		\end{equation*}
		in time  $\tilde{O}((\nnz(A) + d\poly(k/\epsilon))\log(1/\delta))$.
	\end{theorem}
	\begin{proof}
	It is shown in Lemma~\ref{lma:gaussian-lopsided-embedding} that a Gaussian matrix with $O(k)$ rows is a $1/6$-lopsided embedding for $(V_k,A^\mathsf{T})$ with probability $\ge 9/10$. Thus by Lemma~\ref{lma:lopsided-to-epsilon-approximation}, we obtain that
		\begin{equation*}
			\min_{\text{rank-}k\ X}\|A(I-BB^\mathsf{T})\bS^\mathsf{T}X - A(I-BB^\mathsf{T})\|_{1,2} \le (3/2)\text{SubApx}_{k,1}(A(I-BB^\mathsf{T}))
		\end{equation*}
		with probability $\ge 9/10$. \cite{cohen-peng} show that a sampling matrix $\bL$ obtained using Lewis weights has $\tilde{O}(k)$ rows and is a $(1/2,3/2)$ $\ell_1$ subspace embedding for the matrix $A(I-BB^\mathsf{T})\bS^\mathsf{T}$. Thus, the matrices $\bS^\mathsf{T}$ and $\bL$ constructed in Algorithm~\ref{alg:polykapprox} satisfy the conditions of Lemma~\ref{lma:sketch-to-original}. Therefore from Lemma~\ref{lma:sketch-to-actual}, with probability $\ge 3/5$, if a matrix $X$ satisfies
	$
		\|\bL A(I-BB^\mathsf{T})\bS^\mathsf{T}X - \bL A(I-BB^\mathsf{T})\|_{1,2} \le 10\cdot \text{SubApx}_{k,1}(A(I-BB^\mathsf{T})),
	$ then $\|A(I-BB^\mathsf{T})\bS^\mathsf{T}X - A(I-BB^\mathsf{T})\|_{1,2} \le 82 \cdot \text{SubApx}_{k,1}(A(I-BB^\mathsf{T}))$.
	
	Let $\tilde{X} = \argmin_{\text{rank-}k\ X}\|A(I-BB^\mathsf{T})\bS^\mathsf{T} X - A(I-BB^\mathsf{T})\|_{1,2}$. We have $\|A(I-BB^\mathsf{T})\bS^\mathsf{T}\tilde{X} - A(I-BB^\mathsf{T})\|_{1,2} \le (3/2)\text{SubApx}_{k,1}(A(I-BB^\mathsf{T}))$. By Markov's bound, with probability $\ge 3/4$, $\|\bL A(I-BB^\mathsf{T})\bS^\mathsf{T}\tilde{X} - \bL A(I-BB^\mathsf{T})\|_{1,2} \le 10 \cdot \text{SubApx}_{k,1}(A(I-BB^\mathsf{T}))$. We now have the following:
	\begin{equation*}
		\|\bL A(I-BB^\mathsf{T})\bS^\mathsf{T}\tilde{X}(\bL A(I-BB^\mathsf{T}))^+\bL A(I-BB^\mathsf{T}) - \bL A(I-BB^\mathsf{T})\|_{1,2} \le 10 \cdot \text{SubApx}_{k,1}(A(I-BB^\mathsf{T})).
	\end{equation*}
	Thus
	$
		\|A(I-BB^\mathsf{T})\bS^\mathsf{T}\tilde{X}(\bL A(I-BB^\mathsf{T}))^+\bL A(I-BB^\mathsf{T}) - A(I-BB^\mathsf{T})\|_{1,2} \le 82 \cdot \text{SubApx}_{k,1}(A(I-BB^\mathsf{T})).
	$
	Finally,
	\begin{align*}
		&\|A(I-BB^\mathsf{T})(\bL A(I-BB^\mathsf{T}))^+(\bL A(I-BB^\mathsf{T})) - A(I-BB^\mathsf{T})\|_{1,2}\\
		&\le \|A(I-BB^\mathsf{T})\bS^\mathsf{T}\tilde{X}(\bL A(I-BB^\mathsf{T}))^+\bL A(I-BB^\mathsf{T}) - A(I-BB^\mathsf{T})\|_{1,2}\\
		&\le 82\cdot \text{SubApx}_{k,1}(A(I-BB^\mathsf{T})).
	\end{align*}
	The first inequality follows from the fact that for all $x$ and $y$, $\|x^\mathsf{T}(\bL A)^{+}(\bL A) - x^\mathsf{T}\|_2 \le \|y^\mathsf{T}(\bL A)^+(\bL A) - x^\mathsf{T}\|_2$.
	
	By a union bound, with probability $\ge 1/2$, the matrix $\hat{X}$ computed by Algorithm~\ref{alg:polykapprox}, which is an orthonormal basis for the rowspace of $\bL A(I-BB^\mathsf{T})$, satisfies
	\begin{equation*}
		\|A(I-BB^\mathsf{T})(I - \hat{X}\hat{X}^\mathsf{T})\|_{1,2} \le 82 \cdot \text{SubApx}_{k,1}(A(I-BB^\mathsf{T})).
	\end{equation*}
    Thus the matrix $\hat{X}$ which has the minimum value over $\tilde{O}(\log(1/\delta))$ trials satisfies with probability $\ge 1 - \delta$ that
    \begin{equation*}
        \|A(I-BB^\mathsf{T})(I - \hat{X}\hat{X}^\mathsf{T})\|_{1,2} \le O(1) \cdot \text{SubApx}_{k,1}(A(I-BB^\mathsf{T})).
    \end{equation*}
    The running time of Lewis weight sampling can be seen to be $O((\nnz(A)+k^2d(c_1 + k))\log(\log(n)))$ from \cite{cohen-peng}. Thus, the total running time is $\tilde{O}((\nnz(A) + k^2d(c_1 + k))\log(1/\delta))$.
	\end{proof}
\subsection{Finding Best Solution Among Candidate Solutions}
Algorithm~\ref{alg:polykapprox} finds candidate solutions $\hat{X}^{(1)},\ldots,\hat{X}^{(t)}$ for $t = O(\log(1/\delta))$ and returns the best candidate solution that minimizes the cost
\begin{equation}
    \|A(I-BB^\mathsf{T})(I-\hat{X}\hat{X}^\mathsf{T})\|_{1,2}.
\end{equation}
The proof of Theorem~\ref{thm:poly-k-approx} shows that, for all $i=1,\ldots,t$, with probability $\ge 3/5$, $\|A(I-BB^\mathsf{T})(I-\hat{X}^{(i)}(\hat{X}^{(i)})^\mathsf{T})\|_{1,2} \le O(1) \cdot \text{SubApx}_{k,1}(A(I-BB^\mathsf{T}))$. Therefore with probability $\ge 1 - \delta/2$
\begin{equation}
    \min_{i}\|A(I-BB^\mathsf{T})(I-\hat{X}^{(i)}(\hat{X}^{(i)})^\mathsf{T})\|_{1,2} \le O(1) \cdot \text{SubApx}_{k,1}(A(I-BB^\mathsf{T}))
\end{equation}
i.e., with probability $\ge 1 - \delta$, there is a solution $\hat{X}^{(i)}$ among the $t$ potential solutions that has a cost at most $O(1) \cdot \text{SubApx}_{k,1}(A(I-BB^\mathsf{T}))$. We first compute 
\begin{equation*}
    \text{apx}_i = \|A(I-BB^\mathsf{T})(I-\hat{X}^{(i)}(\hat{X}^{(i)})^\mathsf{T})\bG\|_{1,2}
\end{equation*}
where $\bG$ is a scaled Gaussian matrix with $O(\log(n/\delta))$ columns. Values of $\text{apx}_j$ for all $j \in [t]$ can be computed in time $\tilde{O}((\nnz(A) + (n+d)\poly(k/\epsilon))\cdot \log(1/\delta))$. 
We have using the union bound that, with probability $\ge 1 - \delta/2$, for all $j \in [n]$ and $i \in [t]$ that
\begin{equation}
    \|A_{j*}(I-BB^\mathsf{T})(I-\hat{X}^{(i)}(\hat{X}^{(i)})^\mathsf{T})\bG\|_2 = (1/2,3/2) \|A_{j*}(I-BB^\mathsf{T})(I-\hat{X}^{(i)}(\hat{X}^{(i)})^\mathsf{T})\|_2.
\end{equation}
Therefore with probability $\ge 1 - \delta/2$, for all $i \in [t]$,
\begin{equation}
    \apx_i \in (1/2,3/2)\|A(I-BB^\mathsf{T})(I-\hat{X}^{(i)}(\hat{X}^{(i)})^\mathsf{T})\|_{1,2}.
\end{equation}
Let $\tilde{i} = \argmin_{i \in [t]}\apx_i$ and $i^* = \argmin_{i \in [t]}\|A(I-BB^\mathsf{T})(I-\hat{X}^{(i)}(\hat{X}^{(i)})^\mathsf{T})\|_{1,2}$. By a union bound, with probability $\ge 1 - \delta$
\begin{align*}
    \|A(I-BB^\mathsf{T})(I-\hat{X}^{(\tilde{i})}(\hat{X}^{(\tilde{i})})^\mathsf{T})\|_{1,2} &\le 2\apx_{\tilde{i}}\\
    &\le 2\apx_{i^*}\\
    &\le 4\|A(I-BB^\mathsf{T})(I-\hat{X}^{(i^*)}(\hat{X}^{(i^*)})^\mathsf{T})\|_{1,2}\\
    &\le O(1) \cdot \text{SubApx}_{k,1}(A(I-BB^\mathsf{T})).
\end{align*}
Thus, Algorithm~\ref{alg:polykapprox}, with probability $\ge 1 - \delta$, returns a subspace that has cost at most $O(\sqrt{k}) \cdot \text{SubApx}_{k,1}(A(I-BB^\mathsf{T}))$ and has a running time of $\tilde{O}((\nnz(A) + (n+d)\poly(k/\epsilon)) \cdot \log(1/\delta))$.
\subsection{Main Theorem for Constructing a \texorpdfstring{$(1+\varepsilon, k^{3.5}/\varepsilon^2)$}{1+eps,poly(k/eps)} Bicriteria Solution}

  	\begin{theorem}[Residual Sampling]
	    Given matrix $A \in \mathbb{R}^{n \times d}$, matrices $B \in \mathbb{R}^{d\times c_1}$ and $\hat{X} \in \R^{d \times c_2}$ with orthonormal columns such that
$
	        \|A(I-BB^\mathsf{T})(I-\hat{X}\hat{X}^\mathsf{T})\|_{1,2} \le K\cdot \textnormal{SubApx}_{1,k}(A(I-BB^\mathsf{T})),
$
	    Algorithm~\ref{alg:epsapprox} returns a matrix  $U$ having $c = \tilde{O}(c_2 + K \cdot k^3/\epsilon^2\cdot \log(1/\delta))$ orthonormal columns such that with probability $\ge 1 - \delta$
	    \begin{equation}
	        \|A(I-BB^\mathsf{T})(I-UU^\mathsf{T})\|_{1,2} \le (1+\epsilon)\textnormal{SubApx}_{1,k}(A(I-BB^\mathsf{T}))	    \end{equation}
	    in time $\tilde{O}(\nnz(A) + d \cdot \poly(k/\epsilon))$. Moreover we also have that $U^\mathsf{T}B = 0$, i.e., the column spaces of $U$ and $B$ are orthogonal to each other.
	\end{theorem}
	\begin{proof}
	As the matrix $G$ is a Gaussian matrix with $t = O(\log(n/\delta))$ columns, we have that with probability $\ge 1 - (\delta/2)$, for all $i \in [n]$, 
	\begin{equation*}
		\|M_{i*}\|_2 = \|A_{i*}(I-BB^\mathsf{T})(I-\hat{X}\hat{X}^\mathsf{T})G\|_2 = (1 \pm 1/10)\|A_{i*}(I-BB^\mathsf{T})(I-\hat{X}\hat{X}^\mathsf{T})G\|_2. 
	\end{equation*}
	Therefore, the probabilities $p_i$ computed by Algorithm~\ref{alg:epsapprox} are such that
	\begin{equation*}
		p_i = \frac{\|M_{i*}\|_2}{\|M\|_{1,2}} \ge \frac{(9/10)\|A_{i*}(I-BB^\mathsf{T})(I-\hat{X}\hat{X}^\mathsf{T})\|_{2}}{(11/10)\|A(I-BB^\mathsf{T})(I-\hat{X}\hat{X}^\mathsf{T})\|_{1,2}} \ge \frac{9}{11}\frac{\|A_{i*}(I-BB^\mathsf{T})(I-\hat{X}\hat{X}^\mathsf{T})\|_{2}}{\|A(I-BB^\mathsf{T})(I-\hat{X}\hat{X}^\mathsf{T})\|_{1,2}}. 
	\end{equation*} 
	Hence, by applying Lemma~\ref{lma:residual-sampling-main-lemma} to  the matrix $A(I-BB^\mathsf{T})$, we obtain that with probability $\ge 1-\delta$, the matrix $U$ returned by Algorithm~\ref{alg:epsapprox} satisfies
	\begin{equation*}
		\|A(I-BB^\mathsf{T})(I-UU^\mathsf{T})\|_{1,2} \le (1+\epsilon)\text{SubApx}_{1,k}(A(I-BB^\mathsf{T})).
	\end{equation*}
	The matrix $M$ can be computed in time $O(\nnz(A)\log(n/\delta) + (c_1 + c_2)d\log(n/\delta))$. And $s = \tilde{O}(K \cdot k^3/\epsilon^2 \cdot \log(1/\delta))$ independent samples can be drawn from the distribution $p$ in time $O(n + s)$. Finally, the orthonormal basis $U$ can be computed in time $O(d(c+c_1)^2) = O(d\poly(k/\epsilon))$.
	\end{proof}
	
\section{Missing Proofs from Section~\ref{sec:dimensionality-reduction}}
 	\begin{lemma}
		With probability $\ge 2/3$, Algorithm~\ref{alg:dimensionreduction} finds an $\tilde{O}(k^{3}/\epsilon^3)$-dimensional subspace $S$ such that for all $k$-dimensional subspaces $W$,
		\begin{equation*}
			\|A(I-\P_S)\|_{1,2}-\|A(I-\P_{S+W})\|_{1,2}\leq 4\epsilon\cdot\textnormal{SubApx}_{k,1}(A).
		\end{equation*}
		\label{lma:algo-dimreduce}
	\end{lemma}
\begin{proof}
Suppose that the loop in Algorithm~\ref{alg:dimensionreduction} is run for all $t = 10/\epsilon + 1$ iterations instead of stopping after $i^*$ iterations. Let $\hat{X}_i,U_i,B_i$ be the values of the matrices in the algorithm at the end of $i$ iterations. Let $B_0 = []$ be the empty matrix. Condition on the event that all the calls to Algorithm~\ref{alg:polykapprox} in the algorithm succeed. By a union bound over the failure event of each call to Algorithm~\ref{alg:polykapprox}, this event holds with probability $\ge 9/10$. Therefore, by Theorem~\ref{thm:poly-k-approx}, we obtain that
	\begin{align*}
	&	\|A(I-\P_{B_{i-1}})(I - \P_{\hat{X}_i})\|_{1,2}\\
	& \le \tilde{O}(\sqrt{k}) \cdot \text{SubApx}_{k,1}(A(I-\P_{B_{i-1}}))
	\end{align*}
	for all $i \in [10/\epsilon+1]$ and also that $\hat{X}_i$ has $\tilde{O}(k)$ columns. Now we condition on the event that all the calls to Algorithm~\ref{alg:epsapprox} succeed. By a union bound, this holds with probability $\ge 9/10$. Thus we have
	\begin{align*}
		&\|A(I-\P_{B_i})\|_{1,2} = \|A(I-\P_{B_{i-1}})(I - \P_{U_i})\|_{1,2}\\
		&\le (1+\epsilon) \cdot \text{SubApx}_{k,1}(A(I-\P_{B_{i-1}}))
	\end{align*}
	for all iterations $i \in [10/\epsilon+1]$ and also that $U_i$ has $\tilde{O}(k^{3}/\epsilon^2)$ columns which implies that $B_i$ has $\tilde{O}(ik^{3}/\epsilon^2)$ columns. In particular, we have that $\|A(I-\P_{B_1})\|_{1,2} \le (1+\epsilon)\text{SubApx}_{k,1}(A)$. Therefore
\begin{align*}
	&(1+\epsilon)\text{SubApx}_{k,1}(A) - \|A(I-\P_{B_t})\|_{1,2}\\
	&\ge \|A(I-\P_{B_1})\|_{1,2} - \|A(I-\P_{B_t})\|_{1,2}\\
	&=\sum_{i=2}^\mathsf{T} \|A(I-\P_{B_{i-1}})\|_{1,2} - \|A(I-\P_{B_i})\|_{1,2} \ge 0.
\end{align*}
The last inequality follows from the fact that $\text{colspace}(B_i) \supseteq \text{colspace}(B_{i-1})$. The summation in the above equation has $10/\epsilon$ non-negative summands that all sum to at most $(1+\epsilon)\text{SubApx}_{k,1}(A)$. Therefore, at least $9/\epsilon$ summands have value $\le \epsilon(1+\epsilon)\text{SubApx}_{k,1}(A)$. In particular, with probability $\ge 9/10$, 
\begin{equation*}
	\|A(I-\P_{B_{i^*}})\|_{1,2} - \|A(I-\P_{B_{i^*+1}})\|_{1,2} \le \epsilon(1+\epsilon)\text{SubApx}_{k,1}(A).
\end{equation*}
But we also have that
\begin{align*}
	\|A(I-\P_{B_{i^*+1}})\|_{1,2} &= \|A(I - \P_{B_{i^*}})(I-\P_{U_{i^*}})\|_{1,2}\\
	&\le (1+\epsilon)\text{SubApx}_{k,1}(A(I-\P_{B_{i^*}}))\\
	&\le (1+\epsilon)\|A(I - \P_{B_{i^*}})(I-\P_{W})\|_{1,2}\\
	&=(1+\epsilon)\|A(I - \P_{B_{i^*} + W})\|_{1,2}
\end{align*}
where $W$ is any rank $k$ matrix. The second inequality follows from the fact that $\text{SubApx}_{k,1}(A(I-\P_{B_{i^*}})) = \min_{\text{rank-}k\, W}\|A(I-\P_{B_{i^*}})(I-\P_W)\|_{1,2}$. Therefore, for any rank-$k$ matrix $W$, we obtain that
\begin{align*}
	&\|A(I-\P_{B_{i^*}})\|_{1,2} - \|A(I-\P_{B_{i^*} \cup W})\|_{1,2}\\
	&\le \|A(I-\P_{B_{i^*}})\|_{1,2} - \frac{1}{1+\epsilon}\|A(I-\P_{B_{i^*+1}})\|_{1,2}\\
	&\le \|A(I-\P_{B_{i^*}})\|_{1,2} - (1-\epsilon)\|A(I-\P_{B_{i^*+1}})\|_{1,2}\\
	&\le (\|A(I-\P_{B_{i^*}})\|_{1,2} - \|A(I-\P_{B_{i^*+1}})\|_{1,2}) + \epsilon\|A(I-\P_{B_{i^*+1}})\|_{1,2}\\
	&\le 4\epsilon\cdot \text{SubApx}_{k,1}(A).
	\qedhere
\end{align*}
\end{proof}
\begin{theorem}
    Given a matrix $A \in \R^{n \times d}$, $k \in \mathbb{Z}$ and an accuracy parameter $\epsilon > 0$, Algorithm~\ref{alg:complete-dim-reduce} returns a matrix $B$ with $\tilde{O}(k^{3}/\epsilon^6)$ orthonormal columns and a matrix $\text{Apx} = [X\, v]$ such that for any $k$ dimensional shape $S$, 
$
        \sum_i \sqrt{\dist(BX_{i*}^\mathsf{T}, S)^2 + v_i^2} = (1 \pm\varepsilon)\sum_i \dist(A_i, S).
$
    The algorithm runs in time $O(\nnz{(A)}/\epsilon^2 + (n+d)\poly(k/\epsilon))$.
\end{theorem}
\begin{proof}
From the above lemma, we have that the subspace $B$ satisfies with probability $\ge 9/10$, that for any $k$ dimensional subspace $W$,
\begin{equation}
    \|A(I-\P_B)\|_{1,2} - \|A(I-\P_{B \cup W})\|_{1,2} \le \frac{\varepsilon^2}{80}\text{SubApx}_{k,1}(A).
\end{equation}
From Theorem~2.10 of \cite{dw-sketching}, we obtain that with probability $\ge 9/10$, for all $i \in [n]$, the matrix $\bS_j$ found for $i \in [n]$ is such that $\bS_j$ is a $\Theta(\epsilon^2)$ subspace embedding for the matrix $[B\, A_{i*}^\mathsf{T}]$. Therefore, $x_i$ is such that
\begin{equation*}
    \|Bx_{i} - A_{i*}^\mathsf{T}\|_2 \le (1 + \Theta(\epsilon^2))\|(I-BB^\mathsf{T})A_{i*}^\mathsf{T}\|_2.
\end{equation*}
and
$v_i = (1 \pm \Theta(\epsilon^2))\|(I-BB^\mathsf{T})A_{i*}^\mathsf{T}\|_2$. Now the proof follows from Theorem~\ref{thm:approximate-projections-to-approxiamte-distance}.
\end{proof}
	
\section{Missing Proofs from Section~\ref{sec:fastcomputation}}
\subsection{Obtaining an $(O(1),\poly(k))$ Approximation}
	\begin{algorithm}[t]
		\caption{\textsc{PolyApproxDense}}
		\label{alg:polykapproxdense}
		\begin{algorithmic}
		    \STATE {\bfseries Input:} $A\in\R^{n\times d}$, $B \in \R^{d \times c_1}$, $k \in \mathbb{Z}, \delta, b$
			\STATE {\bfseries Output:} $\hat{X} \in \R^{d \times c_2}$
			\STATE $\text{cols} \gets O(k + 1/\delta^2)$
			\STATE $\bS^\mathsf{T} \gets \mathcal{N}(0,1)^{d \times \text{cols}}$
			\STATE $\bW \gets $ {$\ell_1$ embedding for $O(k)$ dimensions from \cite{wang-woodruff}}
			\STATE $[Q,R] \gets $ QR decomposition of $(\bW A)(I-BB^\mathsf{T})\bS^\mathsf{T}$
			\STATE $I_1,\ldots,I_b \gets$ Equal size partition of $[n]$ into $b$ parts
			\STATE $\bC_1 \gets $ Cauchy matrix with $O(\log(n))$ rows
			\FOR {$j = 1,\ldots,b$}
				\STATE $M^{(j)} \gets (\bC_1A_{I_j})(I-BB^\mathsf{T})\bS^\mathsf{T}R^{-1}$
				\STATE $\apx_j \gets \sum_{\text{col} \in \text{cols}(M^{(j)})}\text{median}(\text{abs}(M^{(j)}_{*\text{col}}))$
			\ENDFOR
			\STATE $\bC \gets $ Cauchy matrix with $O(\log(n))$ columns
			\STATE $\text{samples} \gets \tilde{O}(k^{3.5})$
			\STATE $\bL \gets []$
			\FOR {$\text{samples}$ iterations}
				\STATE Sample $j \in [b]$ with probability proportional to $\apx_j$
				\STATE $P^{(j)} \gets A_{I_j}(I-BB^\mathsf{T})\bS^\mathsf{T}R^{-1}\bC$
				\STATE For $i \in I_j$, $p^{(j)}_i \gets \text{median(abs($P^{(j)}_{i*}$))}$
				\STATE Sample $i \in I_j$ with probability proportional to $p^{(j)}_i$
				\STATE Append $\frac{1}{\frac{p_i^{(j)}}{\sum_{i \in I(j)}p_i^{(j)}}\frac{\apx_j}{\sum_{j=1}^b \apx_j}\cdot \text{samples}}e_i^\mathsf{T}$ to matrix $\bL$
			\ENDFOR
			\STATE $\hat{X} \gets $ Orthonormal Basis for rowspace($\bL A(I - BB^\mathsf{T})$)
			\STATE Repeat the above O($\log(1/\delta)$) times and return best $\hat{X}$ 
		\end{algorithmic}
	\end{algorithm}
\begin{theorem}
	Given $A \in \R^{n \times d}$, $B \in \R^{d \times c_1}$, $k \in \mathbb{Z}$ and $\delta$,  Algorithm~\ref{alg:polykapproxdense} returns $\hat{X}$ with $\tilde{O}(k^{3.5})$ orthonormal columns that with probability $1-\delta$ satisfies
		\begin{equation*}
			\|A(I-BB^\mathsf{T})(I - \hat{X}\hat{X}^\mathsf{T})\|_{1,2}
			\le O(1) \cdot \textnormal{SubApx}_{k,1}(A(I-BB^\mathsf{T})).
		\end{equation*}
Given that the matrices $\bC_1A_{I_j}$ for all $j \in [b]$ and $\bW A$ are precomputed for all $O(\log(1/\delta))$ trials, the algorithm can be implemented in time $\tilde{O}(((nd/b)\cdot k^{3.5} + d\poly(k/\epsilon))\log(1/\delta))$.
\end{theorem}
\begin{proof}
	The proof is similar to proof of Theorem~\ref{thm:poly-k-approx}. That proof only makes use of the facts that 
	\begin{enumerate}
		\item for any fixed matrix $M$, $\E[\|\bL M\|_{1,2}] = \E[\|M\|_{1,2}]$,
		\item with probability $\ge 9/10$, for all vectors $x$, $(1/2)\|Ax\|_1 \le \|\bL Ax\|_1 \le (3/2)\|Ax\|_1$, and
	\end{enumerate}
	to conclude with the statement in the theorem. We now show that the matrix $\bL$ computed by Algorithm~\ref{alg:polykapproxdense} satisfies all the above three properties.
	
Note that the random matrix $\bL$ is constructed by sampling $N$ rows, where each row is independently equal to $(1/Np_i)e_i^\mathsf{T}$ with probability $p_i$. Thus
\begin{equation}
    \E[\|\bL M\|_{1,2}] = \E[\sum_{i=1}^N \|\bL_{i*}M\|_2] = N\E[\|\bL_{1*}M\|_2] = N\sum_{j=1}^n\|(1/Np_j)e_j^\mathsf{T}M\|_2p_j = \sum_{j=1}^n \|M_{j*}\|_2 = \|M\|_{1,2}.
\end{equation}
	
	We now prove property 2. From Theorem~1.3 of \cite{wang-woodruff}, we have that $\bW$ has $O(k\log(k))$ rows and that with probability $\ge 99/100$, for all vectors $x$
	\begin{equation*}
		\|A(I-BB^\mathsf{T})\bS^\mathsf{T}x\|_1 \le \|\bW A(I-BB^\mathsf{T})\bS^\mathsf{T}x\|_1 \le O(k\log(k))\|A(I-BB^\mathsf{T})\bS^\mathsf{T}x\|.
	\end{equation*}
	Let $\ell_i = \|A_{i*}(I-BB^\mathsf{T})\bS^\mathsf{T} R^{-1}\|_1$ for $i \in [n]$. From Theorem~\ref{thm:subspace-embedding-to-leverage-score-sampling}, if the probability that the $i^\textnormal{th}$ row is sampled is $\ge (1/2)(\ell_i/\sum_{i'} \ell_{i'})$ for all $i \in [n]$, then the matrix $\bL$ constructed is a $(1/2,3/2)$ $\ell_1$-subspace embedding with probability $\ge 99/100$. Now consider sampling a row of the matrix $\bL$ in the algorithm. We have that the sampled row is in the direction of $e_i$ with probability $(\apx_{j(i)}/\sum_{j' \in [b]}\apx_{j'}) \cdot (p_{i}^{j(i)}/\sum_{i' \in I_{j(i)}}p_{i'}^{j(i)})$. We use $j(i)$ to denote $j \in [b]$ such that $i \in I_j$. We show that this probability is at least $(1/2)(\ell_i/\sum_{i'}\ell_{i'})$. For $j \in [b]$, 
	\begin{equation*}
		\apx_j = \sum_{\text{col}}\text{median}(\text{abs}(M^{(j)}_{*\text{col}})).
	\end{equation*}
 From Theorem 1 of \cite{indyk-l1-estimator}, we have with probability $\ge 1 - 1/100b$ that 
		\begin{equation*}
			\text{median}(\text{abs}(M^{(j)}_{*\text{col}})) = (1 \pm 1/6)\|A_{I_j}(I-BB^\mathsf{T})\bS^\mathsf{T}R^{-1}_{*\text{col}}\|_1.
		\end{equation*}
		Thus $\sum_{\text{col}}\text{median}(\text{abs}(M^{(j)}_{*\text{col}})) = (1 \pm 1/6)\sum_{\text{col}}\|A_{I_j}(I-BB^\mathsf{T})\bS^\mathsf{T}R^{-1}_{*\text{col}}\|_1 = (1 \pm 1/6)\sum_{i \in I_j}\|A_{i*}(I-BB^\mathsf{T})\bS^\mathsf{T}R^{-1}\|_{1} = (1 \pm 1/6)\sum_{i \in I_j}\ell_i$. Therefore, by a union bound, with probability $\ge 99/100$, for all $j \in [b]$
		\begin{equation*}
			\apx_j = (1 \pm 1/6)\sum_{i \in I_j}\ell_i.
		\end{equation*}
Again, from \cite{indyk-l1-estimator}, we obtain that with probability $\ge 99/100$, that for all $i \in [n]$
\begin{equation*}
	\text{median}(\text{abs}(A_{i*}(I-BB^\mathsf{T})\bS^\mathsf{T} R^{-1}\bC)) = (1 \pm 1/6)\|A_{i*}(I-BB^\mathsf{T})\bS^\mathsf{T} R^{-1}\|_1 = (1 \pm 1/6)\ell_i.
\end{equation*}
Thus, with probability $\ge 99/100$, for all $j$ and $i \in I_j$, we have $p^{(j)}_{(i)} = (1 \pm 1/6)\ell_i$. By a union bound, with probability $\ge 98/100$, the probability that an arbitrary row $i$ is sampled in an iteration of the algorithm is 
\begin{equation*}
	(\apx_{j(i)}/\sum_{j' \in [b]}\apx_{j'}) \cdot (p_{i}^{j(i)}/\sum_{i' \in I_{j(i)}}p_{i'}^{j(i)}) \ge \frac{5}{7}\frac{\sum_{i' \in I_j}\ell_{i'}}{\sum_{i'\in [n]}\ell_{i'}}\frac{5}{7}\frac{\ell_i}{\sum_{i'\in I_j}\ell_{i'}} \ge \frac{1}{2}\frac{\ell_i}{\sum_{i'\in [n]}\ell_{i'}}.
\end{equation*}
Thus by a union bound, $\bL$ is a $(1/2,3/2)$ subspace embedding. Now the proof and argument for the running time follow.
\end{proof}
\subsection{Obtaining a $(1+\epsilon,\poly(k/\epsilon))$ Solution}
	\begin{algorithm}
		\caption{\textsc{EpsApproxDense}}
		\label{alg:epsapproxdense}
		\begin{algorithmic}
		    \STATE {\bfseries Input:} $A\in\R^{n\times d}$, $B \in \R^{d \times c_1}$, $\hat{X} \in \R^{d \times c_2}$, $k \in \mathbb{Z}, K, \epsilon, \delta, b$
			\STATE {\bfseries Output:} $U \in \R^{d \times c}$
			\STATE $t \gets O(\log(n))$, $\bG \gets \mathcal{N}(0,1)^{d \times \text{cols}}$
			\STATE $I_1,\ldots,I_b \gets$ Equal size partition of $[n]$ into $b$ parts
			\STATE $\bC_1 \gets $ Cauchy matrix with $O(\log(n))$ rows
			\FOR{$j = 1,\ldots,b$}
				\STATE $M^{(j)} \gets (\bC_1A_{I_j})(I-BB^\mathsf{T})(I-\hat{X}\hat{X}^\mathsf{T})(\bG/t)\sqrt{\pi/2}$
				\STATE $\apx_j \gets \sum_{\text{col} \in \text{cols}(M^{(j)})}\text{median}(\text{abs}(M^{(j)}_{*\text{col}}))$
			\ENDFOR
			\STATE $\text{samples} \gets \tilde{O}(K \cdot k^3/\epsilon^2 \cdot \log(1/\delta))$, $\bS \gets \varnothing$
			\FOR{\textnormal{samples} iterations}
				\STATE Sample $j \in [b]$ with probability proportional to $\apx_j$
				\STATE $P^{(j)} \gets A_{I_j}(I-BB^\mathsf{T})(I-\hat{X}\hat{X}^\mathsf{T})\bG$
				\STATE For $i \in I_j$, $p^{(j)}_i \gets \|P^{(j)}_{i*}\|_2$
				\STATE Sample $i \in I_j$ with probability proportional to $p^{(j)}_i$
				\STATE $\bS \gets \bS \cup i$
			\ENDFOR
			\STATE $U \gets \text{colspan}((I-BB^\mathsf{T})[\hat{X}\, (A_{\bS})^\mathsf{T}])$
			\STATE {\bfseries Return} $U$
		\end{algorithmic}
	\end{algorithm}
\begin{theorem}
	Given a matrix $A \in \R^{n \times d}$, orthonormal matrices $B \in \R^{n \times c_1}$ and $\hat{X} \in \R^{n \times c_2}$ such that
	\begin{equation*}
		\|A(I-BB^\mathsf{T})(I-\hat{X}\hat{X}^\mathsf{T})\|_{1,2} \le K \cdot \text{SubApx}_{1,k}(A(I-BB^\mathsf{T})),
	\end{equation*}
	and parameters $k$, $\epsilon$, and $\delta$, Algorithm~\ref{alg:epsapproxdense} outputs a matrix $U$ with $c = c_1 + \tilde{O}(K \cdot k^3/\epsilon^2 \cdot \log(1/\delta))$ orthonormal columns such that with probability $\ge 1 - \delta$,
	\begin{equation*}
		\|A(I-BB^\mathsf{T})(I-UU^\mathsf{T})\|_{1,2} \le (1+\epsilon)\text{SubApx}_{1,k}(A(I-BB^\mathsf{T})).
	\end{equation*}
	Given that $\bC_1 A_{I_j}$ is precomputed for all $j \in [b]$, the algorithm runs in time $\tilde{O}((nd/b) \cdot (K \cdot k^{3}/\epsilon^2 \log(1/\delta)) + d\poly(k/\epsilon))$.
\end{theorem}
\begin{proof}
We show that the probability that a row $i$ is sampled in an iteration of the Algorithm is $\ge (1/12)\|A_{i*}(I-BB^\mathsf{T})(I-\hat{X}\hat{X}^\mathsf{T})\|_2/\|A(I-BB^\mathsf{T})(I-\hat{X}\hat{X}^\mathsf{T})\|_{1,2}$. Then the proof follows as in the proof of Theorem~\ref{thm:algo-eps-approx-proof}. First assume that $\apx_j$ for $j \in [b]$ computed by the algorithm satisfies
\begin{equation*}
	\apx_j = (1/2,2) \sum_{i \in I_j}\|A_{j*}(I-BB^\mathsf{T})(I-\hat{X}\hat{X}^\mathsf{T})\|_2.
\end{equation*}
Now the probability $p_i$ with which a row $i$ is sampled by the algorithm is given by
\begin{equation*}
	p_i = \frac{\apx_{j(i)}}{\sum_{j \in [b]}\apx_j} \cdot \frac{\|A_{i*}(I-BB^\mathsf{T})(I-\hat{X}\hat{X}^\mathsf{T})\bG\|_2}{\sum_{i' \in I_{j(i)}}\|A_{i'*}(I-BB^\mathsf{T})(I-\hat{X}\hat{X}^\mathsf{T})\bG\|_{2}}.
\end{equation*}
As $\bG$ is a Gaussian matrix with $t = O(\log(n/\delta))$ columns, we have that with probability $\ge 1 - \delta$ that for all $i' \in [n]$ $\|A_{i'*}(I-BB^\mathsf{T})(I-\hat{X}\hat{X}^\mathsf{T})\bG\|_2 = (1 \pm 1/2)\|A_{i'*}(I-BB^\mathsf{T})(I-\hat{X}\hat{X}^\mathsf{T})\|_2 \cdot \sqrt{t}$. Therefore
\begin{align*}
	p_i
	&= \frac{\apx_{j(i)}}{\sum_{j \in [b]}\apx_j} \cdot \frac{\|A_{i*}(I-BB^\mathsf{T})(I-\hat{X}\hat{X}^\mathsf{T})\bG\|_2}{\sum_{i' \in I_{j(i)}}\|A_{i'*}(I-BB^\mathsf{T})(I-\hat{X}\hat{X}^\mathsf{T})\bG\|_{2}}\\
	&\ge \frac{1}{12} \frac{\|A_{i*}(I-BB^\mathsf{T})(I-\hat{X}\hat{X}^\mathsf{T})\|_2}{\|A(I-BB^\mathsf{T})(I-\hat{X}\hat{X}^\mathsf{T})\|_{1,2}}.
\end{align*}
We now prove our assumption which concludes the proof.
	
	Let $x \in \R^d$ be an arbitrary vector. As $\bG$ is a Gaussian matrix with $t = O(\log(n/\delta))$ columns, Lemma~5.3 of \cite{plain-vershynin-one-bit-compressive-sensing} states that
	\begin{equation*}
		\Pr\left[\left|\frac{1}{t}\|x^\mathsf{T}G\|_1 - \sqrt{\frac{2}{\pi}}\|x\|_2\right| \ge \alpha\|x\|_2\right] \le C\exp(-ct\alpha^2).
	\end{equation*}
	Picking an appropriate $\alpha = O(1)$, by a union bound, with probability $\ge 1 - \delta/3$, we obtain
	\begin{align*}
	&	\|{A_{i*}}(I-BB^\mathsf{T})(I-\hat{X}\hat{X}^\mathsf{T})(G/t)\sqrt{\pi/2}\|_1\\
	&= (4/5,6/5)\|A_{i*}(I-BB^\mathsf{T})(I-\hat{X}\hat{X}^\mathsf{T})\|_2
	\end{align*}
	for all $i \in [n]$.
	Now, if $C$ is a Cauchy matrix with $O(\log(n/\delta))$ rows, then with probability $1 - \delta/(3nb)$, we have that
	\begin{equation*}
		\text{median}(\text{abs}(Cx)) = (1 \pm 1/5)\|x\|_1.
	\end{equation*}
	Therefore, by a union bound, we obtain that, with probability $\ge 1 -\delta/3$, for all $j \in [b]$ and $i \in t$ that
	\begin{equation*}
		\text{median}(\text{abs}(CA_{I_j*}(I-BB^\mathsf{T})(I-\hat{X}\hat{X}^\mathsf{T})\bG_{*i})) = (1 \pm 1/5)\|A_{I_j*}(I-BB^\mathsf{T})(I-\hat{X}\hat{X}^\mathsf{T})\bG_{*i}\|_1.
	\end{equation*}
Therefore, with probability $\ge 1 - 2\delta/3$, for all $j \in [b]$,
	\begin{align*}
		\apx_j &= \sum_{i}\textnormal{median}(\textnormal{abs}((M^{(j)})_{*i})) =\sum_{i=1}^\mathsf{T} \text{median}(\text{abs}(CA_{I_j*}(I-BB^\mathsf{T})(I-\hat{X}\hat{X}^\mathsf{T})(\bG_{*i}/t)\sqrt{\pi/2}))\\
		&= (1 \pm 1/5)\sum_{i=1}^\mathsf{T}\|A_{I_j*}(I-BB^\mathsf{T})(I-\hat{X}\hat{X}^\mathsf{T})(\bG_{*i}/t)\sqrt{\pi/2}\|_1\\
		&= (1 \pm 1/5)\sum_{i \in I_j}\|A_{i*}(I-BB^\mathsf{T})(I-\hat X\hat{X}^\mathsf{T})(\bG/t)\sqrt{\pi/2}\|_1\\
		&= (4/5,6/5)(4/5,6/5)\sum_{i \in I_j}\|A_{i*}(I-BB^\mathsf{T})(I-\hat{X}\hat{X}^\mathsf{T})\|_2\\
		&= (1/2,2)\sum_{i \in I_j}\|A_{i*}(I-BB^\mathsf{T})(I-\hat{X}\hat{X}^\mathsf{T})\|_2.
	\end{align*}
The only term in the running time that involves a factor $nd$ is in computing the matrix $P^{(j)}$ for the chosen $j$. A total of $\tilde{O}(K \cdot k^{3}/\varepsilon^2 \cdot \log(1/\delta))$ such $j \in [b]$  are sampled. Therefore, the total running time for computing the matrices $P^{(j)}$ for $j$ sampled by the algorithm is equal to $(nd/b) \cdot \log(n) \cdot \tilde{O}(K \cdot k^{3}/\varepsilon^2 \cdot \log(1/\delta)) + d \cdot \poly(k/\epsilon)$.
\end{proof}
\subsection{Overall Algorithm}
  	\begin{algorithm}[t]
		\caption{\textsc{DimensionReductionDense}}
		\label{alg:dimensionreductiondense}
		\begin{algorithmic}
		    \STATE {\bfseries Input:} $A\in\R^{n\times d}$, $k, \epsilon>0$.
			\STATE {\bfseries Output:} $B \in \R^{d \times c}$ with orthonormal columns
			 \STATE $t \gets 10/\varepsilon + 1$
			 \STATE $i^* \leftarrow $ uniformly random integer from $[10/\epsilon+1]$.
			 \STATE Initialize $B \gets []$
			 \STATE $b \gets k^{3.5}/\epsilon^3$
			 \STATE $\delta = \Theta(\epsilon)$
			\FOR {$i = 1,\ldots,i^*$}
			
			     \STATE $\hat{X} \leftarrow$ \textsc{PolyApproxDense}$(A,B,k,\delta, b)$.
				 \STATE $U \leftarrow$ \textsc{EpsApproxDense}$(A,B,\hat{X},k, \tilde{O}(\sqrt{k}), \Theta(\epsilon),\delta, b)$.
				 \STATE $B \leftarrow [B\, |\, U ]$.
			\ENDFOR
			\STATE {\bfseries Return} $B$.
		\end{algorithmic}
	\end{algorithm} 
\begin{lemma}
Given matrix $A \in \R^{n \times d}, k \in \mathbb{Z}$ and $ \epsilon > 0$, Algorithm~\ref{alg:dimensionreductiondense} returns a matrix $B$ with $\tilde{O}(k^{3.5}/\epsilon^3)$ orthonormal columns such that, with probability $\ge 3/5$, for all $k$ dimensional spaces $W$,
\begin{equation*}
	\|A(I-\P_{B})\|_{1,2} - \|A(I-\P_{B \cup W})\|_{1,2} \le \epsilon \cdot \text{SubApx}_{k,1}(A).
\end{equation*}
The Algorithm runs in time $\tilde{O}(nd + (n+d)\poly(k/\epsilon))$.
\end{lemma}
\begin{proof}
	The proof of the lemma is similar to that of Lemma~\ref{lma:algo-dimreduce}. We now argue that all the pre-computed matrices required across all the iterations of the algorithm can be computed in time $\tilde{O}(nd)$. The Cauchy matrix $\bC_1$ used in Algorithm~\ref{alg:polykapproxdense} has $O(\log(n))$ rows and the matrix $\bW$ has $\tilde{O}(k)$ rows. Note that we have
	\begin{equation*}
		\begin{bmatrix}\bC_1A_{I_1}\\ \bC_1A_{I_2}\\ \\ \bC_1A_{I_b}\\ \bW A\end{bmatrix} = \begin{bmatrix}
 			\bC_1 &  & & & \\
 			 & \bC_1 & & & \\
 			 & & & &\\
 			 &  &  & & \bC_1\\
 			& & \bW & &
 			\end{bmatrix}A.
	\end{equation*}
	Thus all the matrices required for  Algorithm~\ref{alg:polykapproxdense} can be computed by multiplying a $\poly(k/\epsilon) \times n$ matrix with $A$. Similarly, we can compute all the matrices required for Algorithm~\ref{alg:epsapproxdense} by computing the product of a $\poly(k/\epsilon) \times n$ matrix with $A$. Thus, all the matrices required across all iterations of Algorithm~\ref{alg:dimensionreductiondense} can be computed by multiplying a $\poly(k/\epsilon) \times n$  matrix with $A$, which can be done in time $\tilde{O}(nd)$ by the algorithm of \citet{coppersmith}, assuming $n \gg \poly(k/\epsilon)$. Now each iteration of the loop in Algorithm~\ref{alg:dimensionreductiondense} takes $\tilde{O}((nd/b)k^{3.5}/\epsilon^2 + (n+d)\poly(k/\epsilon))$ time. As there are $O(1/\epsilon)$ iterations, the algorithm runs in time $\tilde{O}((nd/b)k^{3.5}/\epsilon^3 + (n+d)\poly(k/\epsilon))$. Since the value of $b$ is chosen to be $k^{3.5}/\epsilon^3$, we obtain that the running time of the algorithm is $\tilde{O}(nd + (n+d)\poly(k/\epsilon))$, including the time to compute the required pre-computed matrices.
\end{proof}
\section{Coreset Construction using Dimensionality Reduction}
Algorithm~\ref{alg:coreset} gives the general algorithm to construct a coreset for any objective involving the sum-of-distances metric. In this section, we discuss the coreset construction for two such problems: the $k$-median and $k$-subspace approximation problems. 
\begin{algorithm}
\caption{\textsc{CoresetConstruction}}
\label{alg:coreset}
\begin{algorithmic}
   \STATE {\bfseries Input:} $A \in \R^{n \times d}$, $k$, $\epsilon$
   \STATE {\bfseries Output:} Coreset
   \STATE $(B, \text{Apx}) \gets \textsc{CompleteDimReduce}(A,k,\epsilon)$
   \STATE Construct a coreset for the instance $\text{Apx}\begin{bmatrix}B^\mathsf{T} & 0\\
   0 & 1\end{bmatrix}$ and return
\end{algorithmic}
\end{algorithm}

For $(B,\text{Apx}=[X\, v])$ returned by Algorithm~\ref{alg:complete-dim-reduce}, we have the guarantee that, with probability $\ge 9/10$, for any $k$-dimensional shape $S$,
\begin{equation*}
	\sum_{i}\sqrt{\dist(BX_{i*}^\mathsf{T}, S)^2 + v_i^2} = (1 \pm \epsilon)\sum_{i}\dist(A_{i*},S).
\end{equation*}
Given a set $S$, let $S_{+1}$ denote the set $\{(s,0)\,|\, s\in S\}$. Let $\text{diag}(B^\mathsf{T}, 1) = \begin{bmatrix}B^\mathsf{T} & 0\\
   0 & 1\end{bmatrix}$. Using this notation, we have that
\begin{equation*}
	\sum_i \dist(\text{Apx}_{i*}\cdot \text{diag}(B^\mathsf{T},1 ), S_{+1}) = (1 \pm \epsilon)\sum_{i}\dist(A_{i*},S).
\end{equation*}
Using the above relation, we give a coreset construction for the $k$-subspace approximation and $k$-median problems. These constructions are as in \cite{sohler2018strong}. For any matrix $M$, let $M_{+1}$ denote the matrix $M$ with a new column of $0$s appended at the end and let $M_{-1}$ denote the matrix $M$ with the last column deleted. 
\begin{theorem}[Coreset for Subspace Approximation]
	There exists a sampling-and-scaling matrix $T$ that samples and scales $\tilde{O}(k^{3}/\epsilon^8)$ rows of the matrix $\text{Apx}$ such that, with probability $\ge 3/5$, for any projection matrix $P$ of rank $k$ that projects onto a subspace $S$ of dimension at most $k$, we have
	\begin{align*}
		\|(( T \cdot \text{Apx} \cdot \text{diag}(B^\mathsf{T},1))_{-1}P)_{+1} - T \cdot \text{Apx} \cdot \text{diag}(B^\mathsf{T},1)\|_{1,2} &= (1 \pm O(\epsilon))\|((\text{Apx} \cdot \text{diag}(B^\mathsf{T},1))_{-1}P)_{+1} - \text{Apx} \cdot \text{diag}(B^\mathsf{T},1)\|_{1,2}\\
		&= (1 \pm O(\epsilon))\sum_i \dist(A_i, S).
	\end{align*}
	This sampling matrix can be computed in time $O(n \cdot \poly(k/\epsilon))$.
\end{theorem}
\begin{proof}
	We first have $\|((\text{Apx} \cdot \text{diag}(B^\mathsf{T},1))_{-1}P)_{+1} - \text{Apx} \cdot \text{diag}(B^\mathsf{T},1)\|_{1,2} = \sum_i \|((\text{Apx}_{i*} \cdot \text{diag}(B^\mathsf{T},1))_{-1}P)_{+1} - \text{Apx}_{i*} \cdot \text{diag}(B^\mathsf{T},1)\|_2 = \sum_i \sqrt{\|(I-P)BX_{i*}^\mathsf{T}\|_2^2 + v_i^2} = \sum_i \sqrt{\dist(BX_{i*}^\mathsf{T}, S)^2 + v_i^2} = (1 \pm \epsilon)\sum_i \dist(A_{i*},S)$. 
	
	We now show $\|(( T \cdot \text{Apx} \cdot \text{diag}(B^\mathsf{T},1))_{-1}P)_{+1} - T \cdot \text{Apx} \cdot \text{diag}(B^\mathsf{T},1)\|_{1,2} = (1 \pm O(\epsilon))\|((\text{Apx} \cdot \text{diag}(B^\mathsf{T},1))_{-1}P)_{+1} - \text{Apx} \cdot \text{diag}(B^\mathsf{T},1)\|_{1,2}$ proving the claim. Let $G$ be a Gaussian matrix with $\tilde{O}(d/\epsilon^2)$ columns. Then with probability $\ge 9/10$, for all $x \in \R^{d+1}$, 
	\begin{equation*}
		\|x^\mathsf{T}G\|_1 = (1 \pm \epsilon)\|x\|_2. 
	\end{equation*}
	See \cite{sohler2018strong} for references. Thus we have that with probability $\ge 9/10$, for all projection matrices $P$ of rank at most $k$, we have
	\begin{equation*}
		\|((\text{Apx} \cdot \text{diag}(B^\mathsf{T},1))_{-1}P)_{+1}G - \text{Apx} \cdot \text{diag}(B^\mathsf{T},1)G\|_{1,1} = (1 \pm \epsilon)\|((\text{Apx} \cdot \text{diag}(B^\mathsf{T},1))_{-1}P)_{+1} - \text{Apx} \cdot \text{diag}(B^\mathsf{T},1)\|_{1,2}.
	\end{equation*}
	Note that for any $P$, the columns of the matrix $((\text{Apx} \cdot \text{diag}(B^\mathsf{T},1))_{-1}P)_{+1}G - \text{Apx} \cdot \text{diag}(B^\mathsf{T},1)G$ lie in the column space of the matrix $\text{Apx}$. Let $T$ be a $(1 \pm \epsilon)$ $\ell_1$-subspace embedding constructed for the matrix $\text{Apx}$ constructed using \cite{cohen-peng}. Therefore
	\begin{equation*}
		\|T \cdot ((\text{Apx} \cdot \text{diag}(B^\mathsf{T},1))_{-1}P)_{+1}G - T \cdot \text{Apx} \cdot \text{diag}(B^\mathsf{T},1)G\|_{1,1} = (1 \pm \epsilon)\|((\text{Apx} \cdot \text{diag}(B^\mathsf{T},1))_{-1}P)_{+1}G - \text{Apx} \cdot \text{diag}(B^\mathsf{T},1)G\|_{1,1}.
	\end{equation*}
	Again, using the fact that $\|x^\mathsf{T}G\|_1 = (1 \pm \epsilon)\|x\|_2$ for all $d+1$ dimensional vectors $x$, we obtain that
	\begin{align*}
		&\|T \cdot ((\text{Apx} \cdot \text{diag}(B^\mathsf{T},1))_{-1}P)_{+1} - T \cdot \text{Apx} \cdot \text{diag}(B^\mathsf{T},1)\|_{1,2}\\
		&= (1 \pm \epsilon)\|T \cdot ((\text{Apx} \cdot \text{diag}(B^\mathsf{T},1))_{-1}P)_{+1}G - T \cdot \text{Apx} \cdot \text{diag}(B^\mathsf{T},1)G\|_{1,1}\\
		&= (1 \pm O(\epsilon))\|((\text{Apx} \cdot \text{diag}(B^\mathsf{T},1))_{-1}P)_{+1}G - \text{Apx} \cdot \text{diag}(B^\mathsf{T},1)G\|_{1,1}\\
		&= (1 \pm O(\epsilon))\|((\text{Apx} \cdot \text{diag}(B^\mathsf{T},1))_{-1}P)_{+1} - \text{Apx} \cdot \text{diag}(B^\mathsf{T},1)\|_{1,2}\\
		&= (1 \pm O(\epsilon))\sum_i \dist(A_i, S).
	\end{align*}
	The matrix $T$ is computed by Lewis Weight Sampling. As the matrix $\text{Apx}$ has dimensions $n \times \tilde{O}(k^{3}/\epsilon^6)$, we see from \cite{cohen-peng} that the matrix $T$ can be computed in time $n \cdot \poly(k/\epsilon)$. 
\end{proof}
\begin{theorem}[Coreset for $k$-median]
	There exists a subset $T \subseteq [n]$ with $|T| = \tilde{O}(k^4/\epsilon^8)$ and weights $w_i$ for $i \in T$ such that, with probability $\ge 3/5$, for any set $C$ of size $k$, 
	\begin{equation*}
		\sum_{i \in T}w_i \dist(\text{Apx}_{i*} \cdot \text{diag}(B^\mathsf{T}, 1), C_{+1}) = (1 \pm \epsilon)\sum_{i \in [n]}\dist(A_{i*}, C).
	\end{equation*}
	Recall that $C_{+1} = \set{(c,0)\,|\,c \in C}$.
\end{theorem}
\begin{proof}
Let $S$ denote the rowspan of the matrix $\text{diag}(B^\mathsf{T},1)$. We have $\text{dim}(S) = \tilde{O}(k^{3}/\epsilon^6)$. Let $\hat{S}$ be the subspace $S$ along with an orthogonal dimension. Thus $\hat{S}$ is an $\tilde{O}(k^{3}/\epsilon^6)$ dimensional subspace of $\R^{d+1}$. Let $C = \set{c_1,\ldots, c_k}$ be an arbitrary set of $k$ centers of $\R^{d+1}$. Now it is easy to see that we can find a set of $k$ points $\hat{C} = \set{\hat{c}_1,\ldots, \hat{c}_k} \subseteq \hat{S}$ such that $\P_S c_i = \P_S \hat{c}_i$ i.e., the projections of $c_i$ and $\hat{c}_i$ onto the subspace $S$ are the same, and also that $\dist(c_i,\P_S(c_i)) = \dist(\hat{c}_i, \P_S(\hat{c}_i))$ and therefore, for any point $a \in S$, $\dist(a, C) = \dist(a, \hat{C})$. 

Now if $T \subseteq [n]$ and the weights $w_i$ for $i \in T$ are such that 
\begin{equation*}
	\sum_{i \in T}w_i\dist(\text{Apx}_{i*} \cdot \text{diag}(B^\mathsf{T}, 1), \tilde{C}) = (1 \pm \epsilon)\sum_{i=1}^n \dist(\text{Apx}_{i*} \cdot \text{diag}(B^\mathsf{T}, 1), \tilde{C})
\end{equation*}
for all $k$-center sets $\tilde{C} \subseteq \hat{S}$, then for any $k$ center set $C \subseteq \R^{d+1}$, we have
\begin{align*}
	\sum_{i \in T} w_i\dist(\text{Apx}_{i*} \cdot \text{diag}(B^\mathsf{T}, 1), C) &= \sum_{i \in T}w_i\dist(\text{Apx}_{i*} \cdot \text{diag}(B^\mathsf{T}, 1), \hat{C})\\
	&= (1 \pm \epsilon)\sum_{i = 1}^n \dist(\text{Apx}_{i*} \cdot \text{diag}(B^\mathsf{T}, 1), \hat{C})\\
	&=(1 \pm \epsilon)\sum_{i = 1}^n \dist(\text{Apx}_{i*} \cdot \text{diag}(B^\mathsf{T}, 1), {C}).
\end{align*}
Thus, preserving the $k$-median distances with respect to the $k$ center sets that lie in $\hat{S}$, preserves the $k$-median distances to all the center sets in $\R^{d+1}$. Using the coreset construction of \citet{feldman-langberg} on the matrix $\text{Apx}$, we can obtain a subset $T \subseteq [n]$ of size $\tilde{O}(k^{4}/\epsilon^8)$ along with weights $w_i$ such that for any $k$-center set $C \subseteq \R^{d+1}$, we have 
\begin{equation*}
	\sum_{i \in T} w_i\dist(\text{Apx}_{i*} \cdot \text{diag}(B^\mathsf{T}, 1), C) = (1 \pm \epsilon)\sum_{i = 1}^n \dist(\text{Apx}_{i*} \cdot \text{diag}(B^\mathsf{T}, 1), {C}).
\end{equation*}
As $\text{Apx}$ is an $n \times \poly(k/\epsilon)$-sized matrix, the algorithm of \citet{feldman-langberg} can be run in time $n \cdot \poly(k/\epsilon)$. Thus, the above subset $T$ and weights $w_i$ for $i \in T$ can be found in time $n\poly(k/\epsilon)$. Now, for any $k$-center set $C \subseteq \R^d$, we have that
\begin{align*}
	\sum_{i = 1}^n \dist(A_{i*}, C) &= (1 \pm \epsilon)\sum_{i=1}^n \sqrt{\dist(BX_{i*}^\mathsf{T}, C) + v_i^2}\\
	&= (1 \pm \epsilon)\sum_{i=1}^n \dist(\text{Apx}_{i*} \cdot \text{diag}(B^\mathsf{T}, 1), C_{+1})\\
	&= (1 \pm \epsilon)\sum_{i \in T} w_i \dist(\text{Apx}_{i*} \cdot \text{diag}(B^\mathsf{T},1), C_{+1}).
\end{align*}
Therefore we obtain a coreset of size $\tilde{O}(k^{4}/\epsilon^8)$ in overall time $\tilde{O}(\nnz(A)/\epsilon^2 + (n+d)\poly(k/\epsilon))$.
\end{proof}
\section{Near-Linear Time Coreset for \texorpdfstring{$k$}{k}-Median}
%
Let $A \in \R^{n \times d}$ be the dataset, where each row $A_{i*}$ of $A$ denotes a point in $\R^d$, for $i \in [n]$. We observe that the coreset construction of \citet{huang-vishnoi} can be implemented in $\tilde{O}(\nnz(A) + (n+d)\poly(k/\epsilon))$ time. The authors only need to compute a constant factor approximation and assignment of each point to a center, which gives a constant factor approximation to the optimum. We show that we can compute such an assignment in time $O(\nnz(A) + (n+d)\poly(k/\epsilon))$.

The usual $k$-median objective is the following
\begin{equation*}
	\min_{y_1,\ldots,y_k \in \R^d}\sum_{i=1}^n \min_j \|A_i^* - y_j\|_2.
\end{equation*}
We can restrict $y_j$ to be a row of $A_i^*$ and lose at most a factor of 2 as follows. Suppose $y_1^*,\ldots,y_k^*$ is the optimal solution. Let $\mathcal{C^*} = (\mathcal{C}_1^*,\mathcal{C}_2^*,\ldots,\mathcal{C}_n^*)$ be the  partition of $[n]$ induced by the optimal solution $y_1^*,\ldots,y_k^*$, where $\mathcal{C}_j^*$ denotes all the indices $i$ such that $y_j^*$ is the closest center to $A_{i*}$. Therefore, the optimal cost for $k$-median is 
\begin{equation*}
	\opt = \sum_{j=1}^k \sum_{i \in \mathcal{C}^*_j} d(A_{i*},y_j^*).
\end{equation*}
Let $A_{c(j)}$ be the point closest to $y_j^*$, i.e.,
\begin{equation*}
	\text{for all $i \in \mathcal{C}_j^*$},\ d(A_{i*},y_j^*) \ge d(A_{c(j)*},y_j^*).
\end{equation*}
We claim that the $k$-median cost of the centers $A_{c(1)},\ldots,A_{c(k)}$ is at most twice the optimum: 
\begin{equation*}
	\sum_{j=1}^k \sum_{i \in \mathcal{C}_j^*}d(A_{i*},A_{c(j)*}) \le \sum_{j=1}^k \left(\sum_{i \in \mathcal{C}_j^*}d(A_{i*},y_j^*) + d(A_{c(j)*},y_j^*)\right) \le \sum_{j=1}^k\sum_{i \in \mathcal{C}_j^*}2d(A_{i*},y_j^*) \le 2\opt.
\end{equation*}
\paragraph{Metric $k$-median} In this version of $k$-median, we restrict to center sets $C$ that are subsets of the data, i.e., we solve the optimization problem
\begin{equation*}
	\min_{y_1,\ldots,y_k \in A}\sum_{i=1}^n \min_j \|A_i^* - y_j\|_2.
\end{equation*}
Let $\opt_{\text{metric}}$ denote the optimum objective value for metric $k$-median. From the above, we obtain that
\begin{equation*}
	\opt_{\text{metric}} \le 2 \text{\opt}.
\end{equation*}
Therefore, a $c$-approximate solution for metric $k$-median is at most a $2c$-approximate solution for Euclidean $k$-median. Let $\Pi$ be a Johnson Lindenstrauss matrix embedding $\R^d$ into $\R^m$, where $m = O(\log(n))$, such that
\begin{equation*}
	\frac{1}{2}d(A_{i*},A_{i'*}) \le d(\Pi A_{i*},\Pi A_{i'*}) \le \frac{3}{2} d(A_{i*}, A_{i'*})
\end{equation*}
for all $i,i' \in [n]$. Now consider the metric $k$-median problem on the points $\Pi A_{1*},\ldots,\Pi A_{n*}$. We can obtain an $11$-approximate solution to the metric $k$-median problem in time $\tilde{O}(nk+k^7)$ (see Theorem~6.2 of \cite{chen2009coresets}). Let $A_{c^*(1)*},\ldots,A_{c^*(k)*}$ be the optimal centers for the metric $k$-median problem on $A_{1*},\ldots, A_{n*}$, and $\Pi A_{c'(1)},\ldots,\Pi A_{c'(k)}$ be an $11$-approximate solution to the metric $k$-median on $\Pi A_{1*},\ldots, \Pi A_{n*}$. Let $\mathcal{C}' = (\mathcal{C}_1,\ldots,\mathcal{C}'_k)$ be the partition of $[n]$ corresponding to this $11$-approximate solution. Then the following shows that $A_{c'(1)},\ldots,A_{c'(k)}$ is a good solution for the metric $k$-median problem on the original dataset:
\begin{align*}
	\sum_{j=1}^k \sum_{i \in \mathcal{C}'_j}d(A_{i*},A_{c'(j)*}) &\le 2\sum_{j=1}^k \sum_{i \in \mathcal{C}'_j}d(\Pi A_{i*},\Pi A_{c'(j)*})\\
	&\le 2 \cdot 11 \sum_{j=1}^k \sum_{i \in \mathcal{C}^*_j}d(\Pi A_{i*},\Pi A_{c^*(j)*})\\
	&\le 2 \cdot 11 \cdot \frac32\sum_{j=1}^k \sum_{i \in \mathcal{C}^*_j}d(A_{i*}, A_{c^*(j)*})\\
	&\le 33\opt_{\text{metric}} \le 66\opt.
\end{align*}

The time taken to compute $\Pi A_{1*},\ldots, \Pi A_{n*}$ is $O(\nnz{(A)}\log(n))$, and then we can compute the $k$ centers and an assignment of points such that this is a $66$-approximate solution in time $\tilde{O}(nk+k^7)$. Using this assignment, we can implement the first stage of importance sampling in the algorithm of \citet{huang-vishnoi} in time $\tilde{O}(\nnz{(A)} + n \cdot \poly(k/\epsilon))$. We note that the first stage of the algorithm of \citet{huang-vishnoi} only needs a constant factor approximation of the distance of a point to its assigned centers, which can be computed as $d(\Pi A_{i*}, \Pi A_{c'(j)*})$, in time $\tilde{O}(\log(n))$, if the point $i$ is assigned to cluster $j$.  The second stage of their algorithm can be implemented in time $d \cdot \poly(k/\epsilon)$. Thus, we can find a strong coreset for $k$-median in time 
\begin{equation*}
	\tilde{O}(\nnz{(A)} + (n+d) \cdot \poly(k/\varepsilon)).
\end{equation*}

\end{document}